\documentclass[12pt]{article}

\usepackage{amsmath,amssymb,amsthm,cite,enumitem,mathtools,xcolor}
\usepackage[margin=3cm]{geometry}

 \usepackage[utf8]{inputenc} 
\usepackage[T1]{fontenc}
\usepackage[osf]{newtxtext} 
\usepackage[nosymbolsc,cmintegrals]{newtxmath}
\usepackage{bm} 
\usepackage[normalem]{ulem}  

\title{Renormalization in string-localized field theories: 
\\
        a microlocal analysis}

\author{Christian Gaß$^{1}$ %
\footnote{Email: cgass@uni-goettingen.de}%
\\[6pt]
{\footnotesize $^1$Institut für Theoretische Physik, 
Georg-August-Universität Göttingen, 37077 Göttingen, Germany}
}

\date{\today}

\DeclareMathOperator{\supp}{supp}   
\DeclareMathOperator{\wf}{WF}
\DeclareMathOperator{\singsupp}{singsupp}
\DeclareMathOperator{\chs}{char}

\newcommand{\eps}{\varepsilon}      
\newcommand{\ka}{\kappa}            
\newcommand{\la}{\lambda}           

\newcommand{\bN}{\mathbb{N}}        
\newcommand{\bR}{{\mathbb{R}}}      
\newcommand{\bS}{\mathbb{S}}        
\newcommand{\mink}{\mathbb{R}^{1+3}} 

\newcommand{\sD}{\mathcal{D}}       
\newcommand{\sS}{\mathcal{S}}       
\newcommand{\sF}{\mathcal{F}}       
\newcommand{\sLi}{\mathcal{L}_\textup{int}} 

\newcommand{\del}{\partial}         

\newcommand{\aside}[1]{} 

\newcommand{\set}[1]{\{\,#1\,\}}  
\newcommand{\vev}[1]{\langle\!\langle#1\rangle\!\rangle} 
\newcommand{\word}[1]{\quad\text{#1}\quad} 
\newcommand{\stringarray}{\boldsymbol{e}} 

\def\wick:#1:{\,\mathopen:#1\mathclose:\,} 

\def\epsi^#1_#2{\eps^{#1}{}_{\!#2}} 
\def\epsii_#1^#2{\eps_{#1}{}^{\!#2}} 
\def\lLa^#1_#2{\Lambda^{#1}{}_{#2}}  

\def\duo<#1,#2>{\langle#1,#2\rangle} 
\def\scal<#1|#2>{\langle#1\mathbin|#2\rangle} 

\theoremstyle{plain}
\newtheorem{thm}{Theorem}           
\numberwithin{thm}{section}         
\newtheorem{lema}[thm]{Lemma}       
\newtheorem{corl}[thm]{Corollary}   

\theoremstyle{definition}
\newtheorem{defn}[thm]{Definition}  
\newtheorem{remk}[thm]{Remark}      
\newtheorem{exap}[thm]{Example}     

\makeatletter
\renewcommand{\section}{\@startsection{section}{1}{\z@}%
                       {-3.5ex \@plus -1ex \@minus -.2ex}%
                       {2.3ex \@plus.2ex}%
                       {\normalfont\large\bfseries}}
\renewcommand{\subsection}{\@startsection{subsection}{2}{\z@}%
                       {-3.25ex \@plus -1ex \@minus -.2ex}%
                       {1.5ex \@plus .2ex}%
                       {\normalfont\normalsize\bfseries}}
\makeatother

\numberwithin{equation}{section}

\usepackage{hyperref}
\usepackage{xcolor}
\hypersetup{
    colorlinks,
    linkcolor={blue!90!black},
    citecolor={blue!90!black},
    urlcolor={blue!90!black}
}

\begin{document}

\maketitle

\medskip

\begin{abstract}
 Using methods of microlocal analysis, we prove that the regularization  
 of divergent amplitudes stays a pure ultraviolet problem in string-localized 
 field theories, despite the weaker localization. Thus, power counting does not 
 lose its significance as an indicator for renormalizability. It also follows 
 that standard techniques can be used to regularize divergent amplitudes 
 in string-localized field theories.
\end{abstract}

\medskip


\section{Introduction}
 \label{sec:intro_slf}
 The foundations of string-localized field theory (SLFT) 
 have been developed thoroughly by Mund, Schroer and Yngvason 
 \cite{MundSY04,MundSY06} in the mid 2000's  and since then, 
 SLFT has been under constant investigation and advancement. 
 At the heart of SLFT is a weaker localization of the 
 potentials of the field strength tensors (for arbitrary 
 masses and spins respectively helicities) along a semi-infinite line 
 referred to as \textit{string}.\footnote{Note that the term string 
 refers to a weaker localization of certain quantum fields and 
 is not  to be confused with the strings of string theory.} 
 These string-localized potentials replace the usual, point-localized, 
 gauge potentials in quantum field theory (QFT) and in that sense, 
 SLFT is not a separate theory. It rather is a different setting 
 within the framework of QFT that exhibits many desirable 
 properties, which we will briefly sketch in the 
 following.

 String-localization has been known for a long time 
 before the works of Mund, Schroer and Yngvason, and has 
 been observed and described more or less explicitly at 
 many occasions in the past:  
 In old works of Jordan \cite{Jordan35} and later also Dirac 
 \cite{Dirac55} on gauge invariant formulations of quantum 
 electrodynamics (QED), the string-localized nature of the dressing 
 factor of the electron field is clearly visible. 
 A derivation of that dressing factor within 
 SLFT as well as an investigation of its consequences on QED 
 has been worked out recently \cite{MundRS21} with the emphasis on 
 the infrared problems of QED.
 
 Also Mandelstam \cite{Mandelstam62} 
 studied QED by employing expressions that 
 clearly resemble the modern string-localized fields. 
 String-locality later reappeared in considerations of Buchholz 
 and Fredenhagen \cite{BF82} and also Steinmann 
 \cite{Steinmann82,Steinmann84}.

 The modern formulation of Mund, Schroer and Yngvason has served 
 as starting point for further investigation of SLFT in the last 
 one-and-a-half decades. It has become clear that string-localized
 fields bear manifold conceptual and practical advantages. First, 
 the string-localized potential for the massless field strength of 
 helicity $s\in\bN$ is a rank-$s$ tensor field that lives on 
 Hilbert space and not on an indefinite Krein space like 
 its point-localized gauge field equivalents \cite{MundSY06,MundRS17b}. 
 Also, string-localization allows for the construction of infinite 
 spin fields \cite{MundSY04}. Moreover, the decoupling of 
 helicities in the massless limit of massive tensor fields is 
 explained by SLFT \cite{MundRS17b}; stress-energy tensors that 
 yield the correct Poincar\'e generators can be constructed 
 for massless fields of arbitrary finite and infinite 
 spin/helicity \cite{MundRS17b,MundRS17a,RehrenPLLimit}, 
 circumventing the Weinberg-Witten theorem \cite{WW80}; 
 the DVZ discontinuity \cite{vDV70,Zh70} 
 in the massless limit of massive gravitons is removed 
 \cite{MundRS17b,MundRS17a}; 
 the Velo-Zwanziger problem \cite{VZ69} has been resolved 
 \cite{Schroer19};
 Gauss' law has been implemented and investigated within 
 SLFT \cite{MundRS20};
 there is no strong CP problem in string-localized QCD \cite{GGBM21}.
 
 To summarize, extensive research on conceptual aspects of SLFT 
 has revealed many benefits. On the other hand, 
 the implementation of string-localized perturbation theory 
 is only in its beginnings. Besides some conceptual considerations 
 \cite{CardosoMV18,MunddeO17}, only computations in low 
 orders and at tree level have been performed. The Lie algebra 
 structure of pure massless Yang-Mills theory and of the weak 
 interaction as well as the chirality of the latter have be 
 derived at second order and tree-level of perturbation 
 theory in a bottom-up approach -- the structure of these 
 interactions is constrained by the requirement that the 
 scattering matrix be string independent \cite{GBMV18,GGBM21}.
 
 Calculations at higher orders of perturbation theory as well as 
 computations of loop graphs involving internal string-localized 
 potentials have not yet been attacked. The main reason for this 
 is the most evident disadvantage of SLFT: The analytic structure 
 of propagators of string-localized potentials is highly complicated. 
 Consequently, an extension of 
 the causal renormalization procedure as described by Epstein and 
 Glaser \cite{EpsteinGlaser73} naively seems very involved and is 
 currently not at hand. In this article, we make a step towards an 
 Epstein-Glaser renormalization scheme in SLFT by proving that 
 the string-localization does actually not affect the singularity 
 structure of the propagators -- provided that care is taken of how 
 the string-localized version of the scattering operator is defined.

 The paper is organized as follows. Section \ref{sec:wf_renorm_slft}
 is a concise introduction to the interrelation of microlocal 
 analysis and renormalization. We also list some basic theorems 
 about wavefront sets, which will be important for our later proofs. 
 The reader familiar with this may skip Section 
 \ref{sec:wf_renorm_slft} and directly proceed to Section 
 \ref{sec:slf}. There, we investigate the distributional nature 
 of string-integration and string-integrated propagators and outline   
 a proper setup of perturbation theory in the string-localized setting.
 Section \ref{sec:WF} contains proofs regarding the existence and 
 extension of products of string-localized propagators. Our results and their 
 connections to other approaches are discussed 
 in Section \ref{sec:discussion}.

 Before continuing, we fix the conventions used in this paper. 
 We employ the mostly negative Minkowski metric 
 $\eta = \text{diag}(1,-1,-1,-1)$. If $x,y$ are Minkowski vectors, 
 we generically denote their Minkowski product by 
 $(xy) := \eta_{\mu\nu} x^\mu y^\nu$ and use $x^2$ for the Minkowski 
 square of $x$. The Fourier transform $\hat{f}(p)=\sF f(p)$ of a 
 function $f(x)$ over Euclidean space $\bR^n$ is defined 
 with negative sign in the exponent, the back transform has a positive sign. 
 All factors of $2\pi$ are absorbed in the back transform. 
 In order to match the physics conventions, the signs of the phase factors 
 in the Fourier 
 transform are inverted over Minkowski space $\mink$ (in addition to   
 the duality pairing being induced by $\eta$). That is,
 \begin{align}
 \label{eq:FT_conventions}
  \hat{f}(p) := \int d^4x \;e^{i(px)} f(x),\quad
  f(x) := \int \frac{d^4p}{(2\pi)^4}\; e^{-i(px)} \hat{f}(p)
 \end{align}
 for a generic $f$ living on $\mink$. When it is relevant, we shall 
 always specify whether statements pertain to $\bR^n$ or $\mink$.
\section{Elements of microlocal analysis needed for renormalization}
\label{sec:wf_renorm_slft}
In the standard approaches to quantum field theory, perturbation 
theory is typically formulated by writing matrix elements of the 
scattering operator as products of numerical distributions 
-- the propagators of the quantum fields involved in a certain model -- 
with the help of Wick's theorem \cite{EpsteinGlaser73}. However, products 
or higher powers of distributions make no sense in general and also 
the products of propagators in the Wick expansion for the scattering 
operator are divergent. At $n$-th order of perturbation theory, they 
only make sense outside the thin diagonal 
$\set{x_1=\cdots=x_n}\subset(\mink)^n$, or after exploiting translation 
invariance, outside the origin $\set{z=0}\subset(\mink)^{n-1}$, 
where $z=(x_1-x_n,\cdots,x_{n-1}-x_n)$.
In momentum space, the non-existence of these products manifests 
itself in the well-known ultraviolet (UV) divergences of loop integrals 
contributing to scattering amplitudes. 
Renormalization in a mathematically rigorous sense is the extension 
of non-existent products of distributions in configuration space 
across the origin $\set{z=0}$ \cite{EpsteinGlaser73,BrunettiF00}.

Once the existence of \textit{some} extension across the origin has 
been established, one must address the question of uniqueness. 
On the one hand, two extensions can only differ by a distribution 
supported at the origin, i.e., by a linear combination of derivatives 
of the Dirac delta, since both extensions must be equal to the 
original distribution outside the origin. On the other hand, adding 
an arbitrary linear combination of derivatives of the Dirac delta to 
a particular extension gives another extension. These ambiguities are 
called \textit{renormalization freedom}. They can be controlled via 
constraints on the short-distance scaling behavior of the extensions, 
i.e., the scaling behavior with respect to $z=0$ \cite{BrunettiF00} 
(cf.~also \cite{Steinmann71}), by requiring that the extension does 
not scale worse than the original distribution. This type of constraint 
is often referred to as \textit{power counting}.

\begin{exap}
 \label{exap:extension_Feynman_prop}
 Consider the massless scalar Feynman propagator $D:=[x^2-i0]^{-1} \in \sS'(\mink)$.
 We will see in Example \ref{exap:product_Feynman_prop} that the square 
 of $D$ is defined on $\mink\setminus0$ but not on the full space 
 $\mink$. For now, we are only interested in constructing an 
 extension. First, note that $D$ is homogeneous, 
 $D(\lambda x) = \lambda^{-2} D(x)$ for all 
 $\lambda>0$. Correspondingly, the square  $(D|_{\mink\setminus0})^2$ scales as 
 $\lambda^{-4}$. Power counting is the requirement that any admissible 
 extension does not scale worse than the non-extended distribution, i.e., 
 one requires that $\lim_{\lambda\downarrow0} \lambda^{4+\omega} w(\lambda x) 
 =0$ for any admissible extension $w$ of $(D|_{\mink\setminus0})^2$ and for all 
 $\omega>0$.
 
 It is a simple task to verify that on $\mink\setminus0$, the square of $D$ 
 coincides with the divergence of the vector-valued distribution 
 \begin{align}
  v^\mu := \frac{1}{2}\frac{x^\mu \,\ln(x^2-i0)}{(x^2-i0)^2}.
 \end{align}
 Since $v^\mu$ is locally integrable with respect to $x$ at $x=0$, it is a 
 well-defined distribution\footnote{The reader may try to verify that 
 the logarithm does not cause any trouble by using the tools that we present 
 in the remaining part of the section.} on the full space $\mink$ and thus, the 
 divergence $\overline{D^2} := \del_\mu v^\mu$ defines an extension of 
 $(D|_{\mink\setminus0})^2$.
 It is also admissible by power counting since 
 $\lim_{\lambda\downarrow0}\lambda^\omega \ln(\lambda^2) =0$ for all $\omega>0$.
 
  An arbitrary extension $w$ of $(D|_{\mink\setminus0})^2$ 
 can only differ from $\overline{D^2}$ by a linear combination of 
 derivatives of the Dirac delta. Power counting introduces an upper 
 bound on the number of derivatives appearing in said linear combination. 
 In the case at hand, 
 \begin{align}
  \label{eq:ren_freedom_exap}
  w- \overline{D^2} = c_0 \delta(x)
 \end{align}
 for some constant $c_0$ and any admissible extension $w$ since the Dirac 
 delta already scales like $\lambda^{-4}$. 
 The free parameter $c_0$ in Eq.~(\refeq{eq:ren_freedom_exap}) 
 introduces a renormalization freedom to the model under consideration. 
 It usually needs to be fixed by physical reasoning. 
 
 The method to obtain the special extension 
 $\overline{D^2}$ is called \textit{differential renormalization} but 
 there are also other well-established methods (see for example 
 \cite{Duetsch18} for an introduction or \cite{BrunettiF00,Dang16}
 for more abstract considerations). 
\end{exap}

\begin{remk}
 The massless Feynman propagator $D$ from Example \ref{exap:extension_Feynman_prop}
 is homogeneous. Therefore, it is obvious how to define the scaling 
 behavior with respect to the origin. A more general definition can for 
 example be found in \cite{BrunettiF00}. 
\end{remk}

A definition of a product of two distributions, which satisfies the known 
rules of calculus, as well as a criterion for its existence was found  
by Hörmander \cite[Theorem 8.2.10.]{Hoermander90} (as a special case of 
Lemma \ref{lema:pullbacks} below). If $u$ and $v$ 
are distributions over an open subset $X\subset\bR^n$, their 
product $uv$ can be defined as the pullback of the tensor product 
$u\otimes v$ by the diagonal map $\Delta: X\rightarrow X\times X$, 
$\Delta(x)=(x,x)$ if 
\begin{align}
\label{eq:hoermander_krit}
 (x;p)\in \wf u \word{implies} (x;-p)\notin \wf v,
\end{align}
where the wavefront set $\wf u$ of a distribution $u$ is a subset of the 
cotangent bundle $\dot{T}^\ast(X)$ over $X$ deprived of the elements $(x;0)$ 
(as indicated by the dot). $\wf u$ gives a refined characterization of  
the singularities of $u$:
\begin{defn}[see Ch.~8 in \cite{Hoermander90}]
 \label{def:singsupp_frequencyset_wf}
 Let $u\in \sD'(X)$ for $X\subset \bR^n$ open. Then the 
 \textit{singular support} $\singsupp u$ of $u$ is the set of points in 
 $X$ that have no open neighborhood where $u$ is smooth. The 
 \textit{frequency set} $\Sigma_x(u)$ of $u$ over a point 
 $x\in X$ is defined as an intersection 
 \begin{align}
 \label{eq:def_frequency_set}
  \Sigma_x(u) := \bigcap_{\substack{\phi\in C_c^\infty(X)\\ \phi(x)\ne0}} 
                 \Sigma(\phi u),
 \end{align}
 where $ \Sigma(\phi u)$ is the cone of directions in $\bR^n\setminus0$ 
 having no conic neighborhood in  
 which the Fourier transform of the compactly supported distribution 
 $\phi u$ is rapidly decaying. Finally, the \textit{wavefront set} 
 $\wf u$ of $u$ is the closed subset of $\dot{T}^\ast(X)$ defined 
 by 
 \begin{align}
  \wf u := \set{ (x;p) \in \dot{T}^\ast(X) \;|\; 
                p \in \Sigma_x(u)}
 \end{align}
 so that the projection of $\wf u$ onto the first component yields 
 the singular support.
\end{defn}
Thus, the wavefront set does not only encode the information 
about the singularities of a distribution but also about the high 
frequencies that are responsible for their appearance. 
It is easy to verify that the wavefront set is a closed and 
conic subset of $\dot{T}^\ast(X)$, where \textit{conic} means that 
the wavefront set is invariant under scaling the second variable 
with positive scalars.

The proofs in Sections \ref{sec:slf} and \ref{sec:WF} will be based on 
several standard statements about properties of the wavefront set. 
For convenience of the reader, we will now concisely list the 
statements on which we will rely later. 

The Hörmander product of two distributions 
$u$ and $v$ is defined as a pullback of their tensor product, provided that 
the criterion (\refeq{eq:hoermander_krit}) is satisfied. One can then 
also give a bound on the wavefront set of the product 
\cite[Theorem 8.2.10.]{Hoermander90}, namely 
 \begin{align}
  \label{eq:bound_wf_product}
  \wf(uv) \subset \set{ (x;p+k) \;|\; (x;p)\in \wf u \text{ or } p=0,\;
                        (x;k)\in \wf v \text{ or } k=0}.
 \end{align}
 The Hörmander product of two distributions is an important special 
 case of the pullback of distributions but we will also need to 
 consider other pullbacks in order to examine the wavefront set 
 of string-localized propagators.
 
 \begin{lema}[Thm.~8.2.4. in \cite{Hoermander90}]
  \label{lema:pullbacks}
 The pullback $f^\ast u$  of a distribution $u \in \sD'(Y)$ 
 by a smooth map $f:X\rightarrow Y$, where $X\subset\bR^m$ and 
 $Y\subset \bR^n$ are open, can be defined such that it coincides with the 
 pullback of smooth maps if $u\in C^\infty(Y)$, provided that 
 $N_f \cap \wf u = \emptyset$, where 
 \begin{align}
  \label{eq:def_Nf}
  N_f := \set{ (f(x);p) \in Y\times \bR^n \;|\; {^t f'}(x) p=0 }.
 \end{align}
 is the set of normals of the map $f$. Moreover, we have 
 \begin{align}
  \label{eq:wf_of_pullback}
  \wf(f^\ast u) \subset f^\ast \wf u 
  := \set{ (x;{^t f'}(x) p) \;|\; (f(x);p) \in \wf u}.
 \end{align}
 \end{lema}

The distributions that appear in quantum field theory are often 
solutions of partial differential equations. For such distributions,  
one can give bounds on their wavefront set:
\begin{lema}[Eq.~(8.1.11) and Thm.~8.3.1. in \cite{Hoermander90}]
 \label{lema:microloc_and_prop_of_sing}
 Let $u\in\sD'(X)$ for $X\subset\bR^n$ open and let 
 $P=\sum_{|\alpha|\le m} a_\alpha(x) \del^\alpha$ be a differential 
 operator of order $m$ on $X$ with smooth coefficients. Then
 \begin{align}
  \wf (Pu) \subset \wf u \subset \wf(Pu) \cup \chs P,
 \end{align}
 where the characteristic set $\chs P$ is defined  in terms of the 
 principal symbol $P_m(x,p) := \sum_{|\alpha| = m} a_\alpha(x) p^\alpha$ 
 of $P$ via 
 \begin{align}
  \chs P := \set{(x;p) \in \dot{T}^\ast(X) \;|\; P_m(x,p)=0}.
 \end{align}
 In particular, if $u$ solves $Pu=0$, then $\wf u \subset \chs P$.
\end{lema}

We will also deal with several homogeneous distributions. 
These are automatically tempered \cite[Theorem 7.1.18.]{Hoermander90}
and the wavefront set of a homogeneous distribution is closely related 
to the wavefront set of its Fourier transform:
 
 \begin{lema}[Thm.~8.1.8. in \cite{Hoermander90}]
  \label{lema:wf_homogeneous} 
  Let $u\in\sD'(\bR^n)$ be homogeneous in $\bR^n\setminus0$. Then 
  \begin{align*}
   (x;p) \in \wf u &\Leftrightarrow (p;-x) \in \wf \hat{u}
   &&\word{if} x\ne0\word{and} p\ne0, \\
   x \in \supp u &\Leftrightarrow (0;-x) \in \wf \hat{u} 
   &&\word{if} x\ne0, \\
   p \in \supp \hat{u} &\Leftrightarrow (0;p) \in \wf u 
   &&\word{if} p\ne0.
  \end{align*}
 \end{lema}
 
 \begin{remk}
  \label{remk:mink_vs_rn}
  The statements from \cite{Hoermander90} displayed in 
  this section are formulated over Euclidean space with 
  the sign convention of the Fourier transform described in 
  the end of the introduction. 
  The mentioned change of the sign convention due to physical reasons
  when working over Minkowski space implies that the covector components 
  of wavefront sets over Minkowski space get an additional sign. 
 \end{remk}
\begin{exap}
 \label{exap:product_Feynman_prop}
 We show that the wavefront set of the massless Feynman propagator 
 $D$ from Example \ref{exap:extension_Feynman_prop} is given by 
 \begin{align}
  \label{eq:wf_DF}
  \wf D = \set{ (x;\lambda x) \;|\; x^2=0,\,x\ne0,\,\lambda >0} 
  \cup \dot{T}^\ast_0.
 \end{align}
 First, we have $\dot{T}^\ast_0 \subset \wf D$ by Lemma 
 \ref{lema:microloc_and_prop_of_sing} since $D$ is a fundamental solution 
 of the wave equation and since $\wf \delta(x) = \dot{T}^\ast_0$.
 The latter wavefront set can be computed by using that 
 $\widehat{\varphi \delta}(p)=\varphi(0)$ for $\varphi\in C_c^\infty(\mink)$. When 
 $x\ne0$, $D$ is the pullback of the homogeneous distribution 
 $[t-i0]^{-1} \in \sS'(\bR)$ by the map $f:\mink\setminus0\rightarrow \bR$ 
 with $f(x) = x^2$. To verify this, note that the Fourier transform 
 of $[t\pm i0]^{-1}$ is a multiple of the Heaviside distribution $\theta(\pm\la)$ 
 and thus, by Lemma \ref{lema:wf_homogeneous}, 
 \begin{align}
 \label{eq:wf_tpm}
 \wf [t\pm i0]^{-1} = \set{(0;\la) \;|\; \la \gtrless 0}
 \end{align}
 and $N_f \cap \wf [t - i0]^{-1}= \emptyset$.
 Hence, the pullback is defined by Lemma \ref{lema:pullbacks}. 
 The wavefront set of the pullback is thus contained in 
 the righthand-side of Eq.~(\refeq{eq:wf_DF}) by Lemma 
 \ref{lema:pullbacks}, where the inverted sign of $\la$ comes from 
 the fact that we work over Minkowski space, as explained in 
 Remark \ref{remk:mink_vs_rn}. Since the wavefront set is 
 conic and the projection onto the first component yields the 
 singular support, $\wf D$ cannot be smaller than the righthand-side 
 of Eq.~(\refeq{eq:wf_DF}).
 
 Since $\lambda$ has a fixed sign, the Hörmander square of $D$ exists 
 when $x\ne0$ but because the wavefront set over $x=0$ contains any 
 direction, the square is not defined at $x=0$. 
\end{exap}
 
 Examples \ref{exap:extension_Feynman_prop} resp.~\ref{exap:product_Feynman_prop} 
 are prototypical for an extension problem in point-localized gauge theories. 
The situation becomes much more complex in string-localized field 
theories. There, the propagators are not only distributions in the variables 
$x$ and $x'$ but also in spacelike string directions $e$ and $e'$. 
The string-localization can induce new singularities to the propagator and 
moreover, the structure of these singularities depends on the formulation 
of a string-localized perturbation theory, as we shall investigate in Section 
\ref{ssc:perturbation_theory}. We will then prove in Section \ref{sec:WF} 
that in a proper setup of string-localized perturbation theory, the 
wavefront sets of string-localized propagators are actually contained 
in the wavefront sets of certain point-localized propagators. That is 
to say, the singularity structure is not worse in SLFT than it is 
in point-localized QFT despite the delocalization. 

To prove the latter statement, another standard theorem from 
microlocal analysis about partially smeared distributions 
will play a central role:
\begin{lema}[Thm.~8.2.12. in \cite{Hoermander90}]
 \label{lema:wf_kernel_smeared}
 Let $X\subset\bR^n$ and $Y\subset \bR^m$ be open and let 
 $K\in \sD'(X\times Y)$ with the corresponding linear transformation 
 $\mathcal{K}$ from $\sD(Y)$ to $\sD'(X)$, i.e., 
 \begin{align}
  [\mathcal{K}\varphi](\phi) = K(\phi\otimes\varphi).
 \end{align}
  Then  
 \begin{align}
  \wf (\mathcal{K}\varphi) \subset 
  \set{(x;p) \;|\; (x,y;p,0) \in \wf K \textup{ for some } y \in \supp \varphi}.
 \end{align}
\end{lema}

\section{String-localized potentials for finite spin/helicity}
\label{sec:slf}
 There is a price to pay for the conceptual advantages of string-localized 
 fields that we have listed in the introduction. String-localized fields 
 do not only depend on the spacetime variable $x$ but also on 
 a spacelike string direction $e\in H$, where $H\subset \mink$ 
 denotes the open subset of spacelike vectors in Minkowski space. 
 In both the massless and the massive case, the string-localized 
 potential $A^{\mu_1\cdots\mu_s}(x,e)$ of the field strength 
 tensor $F_{[\mu_1\nu_1]\cdots[\mu_s\nu_s]}(x)$ for helicity respectively 
 spin $s\in\bN$ can be defined as an $s$-fold integral in string 
 direction via \cite{MundSY06,MundRS17b,MunddeO17}
 \begin{align}
  \label{eq:def_slf_all_heli}
   A_{\mu_1\cdots\mu_s}(x,e) 
   := I_e^s F_{[\mu_1\nu_1]\cdots[\mu_s\nu_s]}(x) 
      e^{\nu_1}\cdots e^{\nu_s}, 
 \end{align}
 where we have introduced the string-integration operator 
 $I_{e} X(x) := \int_0^\infty ds\, X(x+se)$. It is straightforward 
 to verify that $A_{\mu_1\cdots\mu_s}(x,e)$ is indeed a potential 
 for the field strength by exploiting the Bianchi identity for the 
 latter and $e_\mu \del^\mu I_e = -1$.
 By the causal commutation relations for $F_{[\mu_1\nu_1]\cdots[\mu_s\nu_s]}(x)$, 
 two string-localized potentials $A_{\mu_1\cdots\mu_s}(x,e)$ and 
 $A_{\ka_1\cdots\ka_s}(x',e')$ as in 
 Eq.~(\refeq{eq:def_slf_all_heli}) commute if their strings are 
 causally disjoint, i.e., 
 \begin{align}
 \label{eq:slf_CR}
  [A_{\mu_1\cdots\mu_s}(x,e),A_{\ka_1\cdots\ka_s}(x',e')] =0
  \word{if} 
  (x+se-x'-s'e')^2 <0 \;\forall s,s'\ge 0.
 \end{align}
 \begin{remk}
  \label{remk:no_timelike}
   The commutation relations (\refeq{eq:slf_CR}) are in principle 
   also meaningful for lightlike strings and indeed, lightlike 
   string variables have been considered by some \cite{GBMV18}.
   We will discuss in Appendix \ref{app:lightlike_strings}, 
   why they are no reasonable option in our context. 
   Timelike string directions are excluded if the commutation 
   relations (\refeq{eq:slf_CR}) are to remain meaningful.
 \end{remk}

 The integrations 
 in Eq.~(\refeq{eq:def_slf_all_heli}) improve the ultraviolet (UV) 
 scaling behavior of the string-localized potentials in both massive 
 and massless case: They have the same scaling behavior as a scalar 
 field for arbitrary $s\in\bN$. Mund, Schroer and Yngvason \cite{MundSY06} 
 conjectured that this improved UV behavior also has positive 
 effects on renormalizability. However, it is not a priori clear what 
 that means, for string-localized potentials as in Eq.~(\refeq{eq:def_slf_all_heli}) 
 depend not only on the spacetime variable $x$ but also on the string 
 variable $e$ and so do their propagators. Thus, before one can give 
 meaning to the notion of renormalizability, one must answer a few 
 questions: 
 \begin{enumerate}
  \item Of what nature are the products of distributions appearing in a 
        string-localized perturbation theory?
  \item What is the singularity structure of string-localized propagators?
 \end{enumerate}
 We will give answers to these questions in the following.
 
 \subsection{Distributional properties of string-integration}
 \label{ssc:string_int}
 
 In momentum space, string integration as in 
 Eq.~(\refeq{eq:def_slf_all_heli}) becomes a multiplication 
 with factor a factor $ -i [(pe) - i0]^{-1}$ since 
 \begin{align}
   \widehat{I_{e} f}(p)  
  = \int d^4x  \int_0^\infty ds\,  e^{i(p[x-se])} f(x) 
  = \hat{f}(p) \int_0^\infty ds\, e^{- is(pe) }  
  :=  \lim_{\varepsilon\downarrow0}  
  \, \frac{- i \hat{f}(p)}{(pe) - i\varepsilon}.
   \label{eq:pe_first_appearance}
 \end{align}
 Multiplication with such a factor produces additional singularities 
 when $(pe)=0$. Already when setting up their framework for SLFT, 
 Mund, Schroer and 
 Yngvason conjectured that the difficulties coming from these 
 singularities can be cured if the string-localized fields are 
 treated as distributions in both $x$ and $e$ \cite{MundSY06}: 
 ``This opens up the possibility of a perturbative, covariant, 
 implementation of interaction, where the weaker localization 
 (in space-like cones) requires new techniques but promises better 
 UV behavior.'' 
 In Section \ref{sec:WF}, we make a first step towards proving their conjecture
 by showing that the regularization of divergent loop graph amplitudes in SLFT 
 stays a pure short distance problem and that hence the UV scaling 
 behavior remains a meaningful notion.
 
 Let us start our investigations by characterizing the new singularities 
 in detail for general directions $e\in\mink$.  
 \begin{lema}
 \label{lema:Upm_welldef_wf}
  The expressions $U_\pm(p,e) := [(pe) \pm i0]^{-1}$ are tempered 
  distributions on $(\mink)^2$ with 
  \begin{align}\label{eq:wf_upm}
   \wf U_\pm = \set{ (p,e;\la e,\la p) \;|\; 
                \la \lessgtr 0,\; (pe)=0,\;(p,e)\ne (0,0) } 
                \cup \dot{T}^\ast_{(0,0)},
  \end{align}
  where $\dot{T}^\ast_{(0,0)}$ is the cotangent space at $(p,e)=(0,0)$ 
  deprived of the zero-covector.
 \end{lema}

 \begin{proof}
  First note that if $U_\pm$ are well-defined distributions, they are 
  also tempered because they are homogeneous. When $(p,e)\ne(0,0)$, 
  $U_\pm$ are the pullbacks of the distributions 
  $[t \pm i0]^{-1} \in \sS^\prime(\mathbb{R})$ by the map 
  $f: (\mink)^2 \setminus (0,0) \rightarrow \bR$, $f(p,e) = (pe)$
  with set of normals 
  \begin{align}
  N_f = \set{((pe); \la) \in \bR^2 \;|\; 
   \la e = \la p=0,\;(p,e)\ne (0,0) } 
   = \set{(t; 0) \;|\; t \in \bR }
   \end{align}
   so that $N_f\cap \wf [t \pm i0]^{-1}=\emptyset$. Thus, by Lemma
   \ref{lema:pullbacks}, Remark \ref{remk:mink_vs_rn} and the form 
   of $\wf [t \pm i0]^{-1}$ given in Eq.~(\refeq{eq:wf_tpm}), we 
   have
  \begin{align}\label{eq:wf_upm_outside_origin}
   \begin{split}
   \wf \left. U_\pm \right|_{(p,e)\ne(0,0)}
   \subset f^\ast  \wf [t \pm i0]^{-1}
   = \set{ (p,e;\la e,\la p) \;|\; (pe)=0,\; 
            \la \lessgtr 0}.
  \end{split}
  \end{align}
  Eq.~(\ref{eq:wf_upm_outside_origin}) must actually be an equality 
  since the wavefront set is conic and the projection onto the first 
  component must yield the singular support. 
  Since $U_\pm$ are locally integrable at $(p,e)=(0,0)$, we have 
  established their existence as tempered distributions. 
  
  It remains to show that the wavefront set over $(p,e)=(0,0)$ is the 
  whole cotangent space (deprived of the zero-covector). To do so, we 
  introduce the bilinear form 
  \begin{align}
   A := \frac{1}{2} \begin{pmatrix}
         0 & \eta \\
        \eta & 0
        \end{pmatrix}
  \end{align}
  on $(\mink)^2$ such that $A(p,e) = (pe)$ and $4 A^2 = \mathbb{I}$. 
  By \cite[Theorem 6.2.1.]{Hoermander90}, 
  \begin{align}
   (\partial_p \partial_e) \left[ (pe) \pm i0 \right]^{-3} 
   = a_\pm \delta(p,e),
  \end{align}
  where $a_\pm$ are non-vanishing constants that are 
  unimportant for the following arguments. Moreover, we have 
  \begin{align}
   (\partial_p \partial_e)^2 U_\pm  
   = 4 \left[ (pe) \pm i0 \right]^{-3} 
   \quad\Rightarrow\quad 
   (\partial_p \partial_e)^3 U_\pm = 4  a_\pm \delta(p,e).
  \end{align}
  Consequently $\wf{\delta(p,e)} = \dot{T}^\ast_{(0,0)} 
  \subset \wf{U_\pm}$ by Lemma \ref{lema:microloc_and_prop_of_sing} 
  and the proof is completed.  
 \end{proof}
 The distributions $U_\pm$ in Lemma \ref{lema:Upm_welldef_wf} depend 
 on a general string direction $e\in\mink$.  In SLFT, however, the 
 string directions are usually restricted to a set of spacelike 
 (or lightlike) directions. Within our derivations, they are elements of 
 the \textit{open} subset $H\subset\mink$ of spacelike directions, as explained 
 in the beginning of the current section. The restriction of a distribution 
 to an open subset always exists and it follows immediately from  
 Definition \ref{def:singsupp_frequencyset_wf} that the wavefront set 
 of the restricted distribution is the restriction of the wavefront 
 set. We therefore define:
 \begin{defn}
  \label{defn:u_pm}
  Let $u_\pm(p,e) := U_\pm(p,e)|_{\mink\times H}$ denote 
  the restriction of the distributions $U_\pm$ over $(\mink)^2$ 
  from Lemma \ref{lema:Upm_welldef_wf} to the open subset $\mink\times H$ 
  of spacelike string directions with 
   \begin{align}
  \wf u_\pm 
  = \set{ (p,e;x,\xi)\;|\; (p,e;x,\xi) \in \wf U_\pm,\; e\in H }
 \end{align}
 by definition of the wavefront set.
 \end{defn}
  Lemma \ref{lema:Upm_welldef_wf} has the following important consequence for the restricted distributions $u_\pm$.
 \begin{corl}\label{corl:powers_upm}
  Hörmander products $(u_+)^k$ and 
  $(u_-)^k$ of the restrictions to spacelike string 
  variables do exist for arbitrary $k\in\bN$, but the Hörmander product 
  $u_+ \cdot u_-$ with opposite imaginary shift 
  does not exist. Moreover, 
  \begin{align}
   \wf\left[ (u_\pm)^k\right] = \wf u_\pm.
  \end{align}
 \end{corl}
\begin{proof}
  If $e\in H$, then $(p,e)\ne (0,0)$ and 
 \begin{align}
 \label{eq:wf_upm_spacelike}
  \wf u_\pm= \set{ (p,e;\la e,\la p) \;|\; e \in H, 
  \;\la \lessgtr 0,\; (pe)=0}.
 \end{align}
 The Hörmander product of two distributions exists if 
 Eq.~(\refeq{eq:hoermander_krit}) is satisfied. Since the sign of 
 $\la$ in Eq.~(\refeq{eq:wf_upm_spacelike}) is fixed by 
 the sign of the imaginary shift, $(u_\pm)^2$ 
 are defined but $u_+ \cdot u_-$ is not.
 It also follows immediately from the shape of 
 $\wf u_\pm$ and Eq.~(\refeq{eq:bound_wf_product}) that 
 \begin{align}
  \wf \left[ (u_\pm)^2 \right]\subset \wf u_\pm
 \end{align}
 and both sides must be equal since the wavefront set is conic and the 
 projection onto the first component must yield the singular support. 
 By induction, we get the statement for arbitrary powers.
\end{proof}

 \begin{remk}
  \label{remk:closed_spacelike}
   In the literature, the string variables are usually considered as 
   elements of the \textit{closed} subset $H_{-1}\subset\mink$ of 
   spacelike vectors with Minkowski square $e^2 = -1$ 
   (as for example in \cite{MundRS17b,GGBM21}). The restriction 
   of a distribution to a closed subset is much more involved than 
   the restriction to an open subset. It does not always exist and 
   even if it does, it may affect the form of the wavefront set 
   \cite{Hoermander90}. We will briefly sketch in Appendix 
   \ref{app:H_minus_1} why the restriction to $H_{-1}$ is indeed 
   unproblematic. For our purposes, however, the simpler case of the 
   restriction to the open subset $H$ is sufficient. 
 \end{remk}

Lemma \ref{lema:Upm_welldef_wf} and Corollary \ref{corl:powers_upm} 
are the starting point for the full analysis of the singularities 
of string-localized propagators that we will subsequently perform.

\subsection{String-localized propagators for all spins and helicities}
\label{ssc:slf_props}
 For two arbitrary point-localized fields $X(x)$ and $X'(x')$ of mass 
 $m\ge0$, we introduce the notation 
 \begin{align}
  \label{eq:generic_2pf_pl}
  \vev{X(x) X'(x')} := \int d\mu_m(p)\, e^{-i(p(x-x'))} \,{_mM}^{X,X'}(p)
 \end{align}
 for the two-point function of $X$ and $X'$, with the measure 
 $d\mu_m(p) = \frac{d^4p}{(2\pi)^3} \delta(p^2-m^2) \theta(p^0)$  
 on the mass shell and where ${_mM}^{X,X'}(p)$ is a polynomial 
 in $p$. Furthermore, we write 
 \begin{align}
  \label{eq:kinematic_prop_pl}
  \vev{T_0 X(x) X'(x')} := \int \frac{d^4p}{(2\pi)^4}\, 
  \frac{e^{-i(p(x-x'))}}{p^2-m^2 +i0} \,{_mM}^{X,X'}(p)
 \end{align}
 for the corresponding \textit{kinematic} propagator. Using translation 
 invariance $x-x'\rightarrow x$, we sometimes also use the notation 
 \begin{align}
  \vev{X X'}(x) \word{resp.} \vev{T_0 X X'}(x).
 \end{align}

 The kinematic propagator from Eq.~(\refeq{eq:kinematic_prop_pl}) 
 is in general only a specific choice for a propagator since the transition 
 from Eq.~(\refeq{eq:generic_2pf_pl}) to Eq.~(\refeq{eq:kinematic_prop_pl}) 
 might be non-unique: 
 Dependent on the scaling behavior of 
 ${_mM}^{X,X'}(p)$, there can arise ambiguities in the definition of 
 time-ordering at $x=x'$ \cite{EpsteinGlaser73,ScharfLast}. 
 Any other propagator can only differ from 
 Eq.~(\refeq{eq:kinematic_prop_pl}) by a linear combination 
 \begin{align}
  \sum_{|\alpha|\le n} b_\alpha \del^\alpha \delta(x-x'),
  \label{eq:freedom_T_pl_xspace}
 \end{align}
 where $\alpha$ is a 
 multi-index, $b_\alpha$ are constants and $n \in \bN_0$ is restricted 
 by the scaling behavior in a similar manner to the restrictions 
 from power counting displayed in Example \ref{exap:extension_Feynman_prop}. The ambiguity 
 (\refeq{eq:freedom_T_pl_xspace}) can be understood 
 more easily by a momentum space consideration. Adding a term 
 \begin{align}
 \label{eq:freedom_T_pl_pspace}
  (p^2-m^2) \tilde{M}(p)
 \end{align}
 to ${_mM}^{X,X'}(p)$, where $\tilde{M}(p)$ is another polynomial, 
 does not contribute to the two-point function 
 (\refeq{eq:generic_2pf_pl}) but yields a contribution of the form 
 (\refeq{eq:freedom_T_pl_xspace}) to the propagator 
 (\refeq{eq:kinematic_prop_pl}).
 
 \begin{remk}
  We will frequently refer to the expression ${_mM}^{X,X'}(p)$ as 
  \textit{kernel} of a propagator or two-point function and hope 
  that this usage does not cause confusion with distribution kernels 
  that will be used implicitly in Section \ref{sec:WF}.
 \end{remk}

  String-integrating $X(x)$ in Eq.~(\refeq{eq:generic_2pf_pl}) gives an 
 additional factor $-iu_-(p,e)$ in momentum space, while string 
 integrating $X'(x')$ yields a factor $iu_+(p,e')=-iu_-(p,-e')$. 
 A natural choice of a propagator involving a string-integrated field 
 is thus given by inserting the appropriate powers of $-iu_-(p,e)$ and 
 $iu_+(p,e')$ into Eq.~(\refeq{eq:kinematic_prop_pl}). Again, the 
 propagator might not be unique but two propagators can at most differ 
 by a linear combination of string-integrated Dirac deltas. 
 We will investigate these ambiguities in Section 
 \ref{ssc:non_kinematic_propagators}. For now, we prove the 
 well-definedness of the relevant class of momentum space 
 representations of string-localized \textit{kinematic} propagators.\footnote{There 
 are similar statements for the two-point functions 
 but in this paper, we are interested in scattering theory only. 
 Therefore, we only consider propagators. The interested 
 reader may carry through the existence proof for the two-point 
 functions as an exercise, using our findings as a guide.}

 \begin{lema}[massless case]
  \label{lema:well_def_SI_2pf_and_prop_massless}
  Let $m=0$ and let $M_\times(p)$ be a polynomial in $p$ such that 
  $\omega\in\mathbb{N}_0$ is the smallest power of $p$ appearing in 
  $M_\times$, where the subscript $\times$ is a 
  placeholder for possible Lorentz indices. Let further 
  $k,k' \in \mathbb{N}_0$ and $\omega-k-k'-2>-4$. Then the expression
  \begin{align}
  \label{eq:massless_kernel_generic}
  \frac{[u_-(p,e)]^k [u_+(p,e')]^{k'} M_\times(p)}{p^2+i0}
  \end{align}
  is a well-defined (and possibly tensor-valued) distribution on 
  $\mink\times H^2$.
 \end{lema}

 \begin{proof}
  By Corollary \ref{corl:powers_upm}, the powers $[u_-(p,e)]^k$ and 
  $[u_+(p,e')]^{k'}$ exist on $\mink\times H$ and their wavefront set 
  is given by Eq.~(\refeq{eq:wf_upm_spacelike}). We promote them 
  to distributions on $\mink\times H^2$ by tensoring with the 
  constant distribution in the missing string variable, so that 
  \begin{subequations}
  \begin{align}
  \label{eq:wf_upm_otimes_prime}
   \wf \left( [u_-(p,e)]^k \otimes 1_{e'} \right)
   &= \set{ (p,e,e';\la e,\la p,0) \;|\; \la >0 ,\; (pe)=0},
   \\
   \label{eq:wf_upm_otimes}
   \wf \left( [u_+(p,e')]^{k'} \otimes 1_e \right)
   &= \set{ (p,e,e';\ka e',0, \ka p) \;|\; \ka <0 ,\; (pe')=0}.
  \end{align}
  \end{subequations}
  By Eq.~(\refeq{eq:wf_DF}) and Lemma \ref{lema:wf_homogeneous}, 
  \begin{align}
       \wf \left([p^2+i0]^{-1} \otimes 1_{e,e'}\right) 
    = \; &\set{ (p,e,e';\la p,0,0) \;|\; p^2=0,\; p\ne0,\;\la<0 } 
      \notag \\      \cup 
      &\set{(0,e,e';x,0,0)  \;|\; x\in \mink\setminus0}.
      \label{eq:wf_p2_inverse}
  \end{align}
  Hence, the covector components of the three wavefront sets
  (\refeq{eq:wf_upm_otimes_prime}), (\refeq{eq:wf_upm_otimes}) and 
  (\refeq{eq:wf_p2_inverse}) cannot add up to zero when $p\ne0$ 
  and thus the Hörmander product exists on 
  $(\mink\setminus0)\times H^2$. The requirement that $\omega-k-k'-2>-4$ 
  ensures that (\refeq{eq:massless_kernel_generic}) is locally 
  integrable at $p=0$. Therefore it is a well-defined distribution 
  on $\mink\times H^2$.
 \end{proof}
  
  Lemma \ref{lema:well_def_SI_2pf_and_prop_massless} has an 
  analogue for $m>0$, where a weaker constraint on the smallest 
  power of $p$ appearing in the polynomial $M_\times(p)$ is 
  sufficient to guarantee local integrability at $p=0$ because  
  the denominator $[p^2-m^2+i0]^{-1}$ has a better behaved 
  wavefront set for $m>0$ than for $m=0$.
  
  \begin{lema}[massive case]
  \label{lema:well_def_SI_2pf_and_prop_massive}
  Let $m>0$ and let $M_\times(p)$ be a polynomial in $p$ such that 
  $\omega\in\mathbb{N}_0$ is the smallest power of $p$ appearing in 
  $M_\times$, where the subscript $\times$ 
  is a placeholder for possible Lorentz indices. Let further 
  $k,k' \in \mathbb{N}_0$ and $\omega-k-k'>-4$. Then the 
  expression
  \begin{align}
  \label{eq:massive_kernel_generic}
  \frac{[u_-(p,e)]^k [u_+(p,e')]^{k'} M_\times(p)}{p^2-m^2+i0}
  \end{align}
  is a well-defined (and possibly tensor-valued) distribution on $\mink\times H^2$.
 \end{lema}
 \begin{proof}
  The numerator of (\refeq{eq:massive_kernel_generic}) is the same as 
  in Lemma \ref{lema:well_def_SI_2pf_and_prop_massless}. Since 
  $[p^2-m^2+i0]^{-1}$ is smooth at $p=0$ for $m>0$, it is enough to 
  require $\omega-k-k'>-4$ for local integrability at $p=0$  
  instead of $\omega-k-k'-2>-4$, which was necessary in the massless 
  case. 
  
  Similar to the procedure in the proof of Lemma \ref{lema:Upm_welldef_wf},
  the distribution $[p^2-m^2+i0]^{-1}$ can be seen as the pullback 
  of $[t+i0]^{-1}\in\sS'(\bR)$ by the map $f:\mink \rightarrow \bR$
  with $f(p) = p^2-m^2$ with set of normals
  \begin{align}
   N_f = \set{ (p^2-m^2;\xi) \in \bR^2 \;|\; \xi p=0}
  \end{align}
  because $N_f \cap \wf [t+i0]^{-1} = \emptyset$. Then 
  \begin{align}
   \wf [p^2-m^2+i0]^{-1} \subset f^\ast \wf [t+i0]^{-1} 
   = \set{ (p;\la p) \;|\; p^2=m^2,\; \la < 0},
  \end{align}
  where the inverted sign of $\la$ again comes from the 
  fact that we work over Minkowski space, as explained in 
  Remark \ref{remk:mink_vs_rn}.
  Clearly, the covector components of 
  $\wf\left( [p^2-m^2+i0]^{-1} \otimes 1_{e,e'} \right)$ 
  and the wavefront sets  (\refeq{eq:wf_upm_otimes_prime}), 
  (\refeq{eq:wf_upm_otimes}) cannot add up to zero when $p\ne0$, 
  giving the well-definedness of (\refeq{eq:massive_kernel_generic})
  as a distribution on $\mink \times H^2$.
 \end{proof}
 
 \begin{remk}
  The conditions $\omega-k-k'-2>-4$ in the massless and 
  $\omega-k-k'>-4$ in the massive case ensure local integrability 
  with respect to $p=0$. It does not automatically follow that 
  the respective distributions are ill-defined if these integrability 
  conditions are not satisfied. However, when investigating 
  the position space representation of the doubly string-integrated 
  massless Feynman propagator $I_e I_{-e'} [x^2-i0]^{-1}$, the 
  singularity is explicitly observed as an infrared effect 
  \cite{GRT21}. 
 \end{remk}

 \begin{remk}
  In configuration space, one might be tempted to circumvent the 
  integrability conditions in Lemmas 
  \ref{lema:well_def_SI_2pf_and_prop_massless} and 
  \ref{lema:well_def_SI_2pf_and_prop_massive} by shifting the 
  string-integration operation to the $x$-part of the test function. 
  Since the latter is a Schwartz function, application of any finite number 
  of string-integrations to it remains finite. However, the result 
  is no Schwartz function anymore. In direction of the string, it 
  converges to a constant that is in general non-zero. Therefore, 
  the integrability conditions are necessary if one does not wish to 
  leave the regime of distribution theory. 
 \end{remk}

 The string-localized potentials defined in 
 Eq.~(\refeq{eq:def_slf_all_heli}) are homogeneous of degree $0$ 
 in the string variable because each string-integration $I_e$ is 
 accompanied by a factor $e^\mu$. Therefore it proves useful to 
 replace the distributions $u_\pm(p,e)$ by the vector-valued 
 distributions 
 \begin{align}
 \label{eq:def_q_pm}
 q_\pm^\mu (p,e) := \pm i u_\pm(p,e) \,e^\mu.
 \end{align}
 The wavefront set of a vector-valued distribution is defined as the 
 union of the wavefront sets of the components. But since each component 
 of $q^\mu_\pm$ is nothing but $u_\pm$ times a smooth function, 
 we have 
 \begin{align}
  \label{eq:wfq_contained_in_wfu}
  \wf q_\pm^\mu \subset \wf u_\pm.
 \end{align}
 As long as $e\ne0$, in particular if $e$ is spacelike, both wavefront 
 sets in Eq.~(\refeq{eq:wfq_contained_in_wfu}) are equal. Thus, if 
 \begin{align}
  {_mM}^{F,F}_{[\mu_1\nu_1]\cdots[\mu_s\nu_s]
            [\ka_1\la_1]\cdots[\ka_s\la_s]}(p)
 \end{align}
 denotes the kernel of the kinematic spin/helicity-$s$ field strength 
 propagator, then the kinematic propagator for the string-localized 
 potential is given by replacing ${_mM}^{F,F}$ by  
 \begin{align}
  {_mM}^{A,A}_{\mu_1\cdots\mu_s\ka_1\cdots\ka_s}(p,e,e') 
  := 
  {_mM}^{F,F}_{[\mu_1\nu_1]\cdots[\mu_s\nu_s]
            [\ka_1\la_1]\cdots[\ka_s\la_s]}
 \prod_{i=1}^s q_-^{\nu_i}(p,e) q_+^{\la_i}(p,e')
 \end{align}
 in Eq.~(\refeq{eq:kinematic_prop_pl}), provided that the requirements 
 of Lemma~\ref{lema:well_def_SI_2pf_and_prop_massless} or 
 \ref{lema:well_def_SI_2pf_and_prop_massive}, respectively, are 
 satisfied (cf.~also \cite{MunddeO17}).

   For both $m>0$ and $m=0$, the kernel ${_mM}^{F,F}_{\mu\nu,\ka\la}$ 
   of field strength $F_{\mu\nu}$
  of spin/helicity $s=1$ reads
  \begin{align}
   {_mM}^{F,F}_{\mu\nu,\ka\la}(p) = -\eta_{\mu\ka} p_\nu p_\la 
   + \eta_{\mu\la} p_\nu p_\ka + \eta_{\nu\ka} p_\mu p_\la 
   -\eta_{\nu\la} p_\mu p_\ka
  \end{align}
  and is homogeneous of degree $\omega=2$. Since for $s=1$, the 
  string-localized potential is $A_\mu(x,e)=I_e F_{\mu\nu}(x)e^\nu$, 
  string-integration of both $F$'s in the propagator gives 
  a total factor $q_-^\nu(p,e)\cdot q_+^\la(p,e')$ so that 
  $\omega-k-k'=0$. Then the requirements for the respective 
  Lemmas~\ref{lema:well_def_SI_2pf_and_prop_massless} and 
  \ref{lema:well_def_SI_2pf_and_prop_massive} are met and 
  \begin{align}
    {_mM}^{A,A}_{\mu\ka}(p,e,e') 
    &= -\eta_{\mu\ka} 
        + \frac{e_\ka p_\mu}{(pe)-i0} 
        + \frac{{e'}_\mu p_\ka}{(pe')+i0} 
        - \frac{(ee') p_\mu p_\ka}{[(pe)-i0][(pe')+i0]} 
    \notag \\
    &=: - E_{\mu\ka}(p,e,e').
  \label{eq:def_Emunu}
  \end{align}
  
It turns out that the quantity $E_{\mu\kappa}(p,e,e')$ from 
Eq.~(\refeq{eq:def_Emunu}), which is the kernel for the kinematic
propagator of $s=1$ string-localized potentials is enough to
describe the kernel of the kinematic propagator for string-localized 
potentials of arbitrary spin/helicity \cite{MundRS17b}. By symmetry 
of the field strengths and therefore also of the potentials 
in Eq.~(\refeq{eq:def_slf_all_heli}), one does not lose any information 
if one contracts all indices of $A_{\mu_1\cdots\mu_s}(x,e)$ with the 
same (arbitrary) four-vector $f^\mu$ and defines 
\begin{align}
\label{eq:def_A_f}
 A_f^{(s)}(x,e) := f^{\mu_1}\cdots f^{\mu_s} A_{\mu_1\cdots\mu_s}(x,e).
\end{align}
The authors of \cite{MundRS17b} were able to prove -- without considering 
questions of well-definedness -- that the kernel of the two-point function 
of $A^{(s)}(x,f)$ and $A^{(s)}(x',f')$ for all $m\ge0$ is given by
\begin{align}
 \label{eq:kernel_general_s}
 {_mM^{A^{(s)}_f,A^{(s)}_{f'}}}(p,e,e') 
 = (-1)^s \sum_{2n\le s} \beta_n^s \, (E_{ff})^n ( E_{f'f'})^n  (E_{ff'})^{s-2n}
 ,
\end{align}
 with coefficients $\beta_n^s$ that are of no interest here. 
 We have used the notation  
\begin{align}
 E_{ff} &:=  f^{\mu} E_{\mu\nu}(p,e,-e)f^{\nu},\quad  
 E_{f'f'} := {f'}^{\mu} E_{\mu\nu}(p,-e',e')  {f'}^{\nu},\word{and} \notag \\
 E_{ff'} &:= f^{\mu} E_{\mu\nu}(p,e,e')  {f'}^{\nu}.
  \label{eq:E_ff_etc}
\end{align}
The signs in the arguments of Eq.~(\refeq{eq:E_ff_etc}) ensure that 
each $(pe)$ is accompanied by a shift $-i0$ and each $(pe')$ is 
accompanied by a shift $+i0$. Moreover, it is clear that the 
kernels from Eq.~(\refeq{eq:kernel_general_s}) arise from an $s$-fold 
$e$-integration and an $s$-fold $-e'$-integration times a polynomial 
in $p$ which is homogeneous of degree $\omega=2s$. Therefore, 
Lemmas \ref{lema:well_def_SI_2pf_and_prop_massless} and 
\ref{lema:well_def_SI_2pf_and_prop_massive} apply to the kinematic 
propagators with kernels as in Eq.~(\refeq{eq:kernel_general_s}) for 
all $s$ and have the following Theorem as a corollary. 

\begin{thm}
 \label{thm:well_def_kinematic_props_all_s}
 The kinematic propagators of the string-localized potentials 
 for arbitrary mass $m\ge0$ and spin/helicity $s\in\bN$ defined 
 by
 \begin{align}
  \vev{T_0 A^{(s)}_f A^{(s)}_{f'}}(x) 
  :=  \int \frac{d^4p}{(2\pi)^4} \frac{e^{-i(px)} }{p^2-m^2+i0} 
   \;{_mM^{A^{(s)}_f,A^{(s)}_{f'}}}(p,e,e')
 \end{align}
 are well-defined distributions on $\mink \times H^2$.
\end{thm}

\subsection{Perturbation theory with string-localized fields}
\label{ssc:perturbation_theory}
At the heart of perturbation theory in quantum field theory is the 
construction of the scattering operator, or S-matrix $\bS$ as a formal 
power series (Dyson series) of time-ordered products of an interaction 
Lagrangian $\sLi$ describing a certain model \cite{BogoliubovS59}.

In the usual point-localized theories, $\bS=\bS[g]$ is considered as a 
functional of a multiplet of Schwartz functions $g$, which are 
interpreted as large distance cutoffs of coupling constants, and tend to 
constants in the adiabatic limit. For 
a rigorous construction, one usually axiomatizes certain properties of 
the S-matrix, such as its unitarity, Lorentz and translation invariance 
and causality \cite{EpsteinGlaser73,BogoliubovS59}. Additionally, it can 
be subject to (sometimes model-dependent) internal and discrete symmetries 
\cite[Sec. 3.3]{Weinberg95I}. In gauge theories, its form is further 
constrained by the requirement of perturbative gauge invariance 
\cite{AsteS99,DuetschS99,ScharfLast}. There is no concept of gauge in 
string-localized field theories and the gauge invariance principle 
is replaced by the requirement for string independence of the 
S-matrix (see for example \cite{GBMV18,GGBM21}).

The time-ordered products of operator-valued distributions (i.e., of 
$\sLi$) that appear in the Dyson series for $\bS[g]$ are usually 
reduced to products of numerical distributions by taking expectation 
values and employing Wick's theorem. These products are ill-defined 
at a diagonal set and renormalization is the extension of these 
products of distributions to the whole space 
\cite{EpsteinGlaser73,BrunettiF00}, as we have illustrated 
in Example \ref{exap:extension_Feynman_prop}. 
In the following, we sketch a possible transition of these 
notions, which are well-established in point-localized theories, 
to SLFT. 

 As a first step for this transition, one must declare the nature of 
 the string-localization. Is string-localization a feature of the 
 potentials, the Lagrangian or the S-matrix? That is to say: Does 
 each field come with its own string variable, do the fields in 
 the interaction Lagrangian depend on the same string variable or do 
 \textit{all} appearing fields depend on the same string variable: 
 \begin{subequations}
  \begin{align}
   \label{eq:alternative_A}
   \sLi &=\sLi(x,e_1,\cdots,e_k) &&\word{(each SL field has its own string variable),}\\
    \label{eq:alternative_L}
   \sLi &= \sLi(x,e) &&\word{(all SL fields in $\sLi$ depend on the same $e$),} \\
   \label{eq:alternative_S}
   \bS &= \bS[g;e]  &&\word{(there is only a single string variable).}
  \end{align}
 \end{subequations}
 In a generic model, the three alternatives (\refeq{eq:alternative_A}), 
 (\refeq{eq:alternative_L}) and (\refeq{eq:alternative_S}) result in 
 completely different analytic properties of the corresponding 
 perturbation theory. Note, however, that the alternatives 
 (\refeq{eq:alternative_A}) and (\refeq{eq:alternative_L}) are 
 equivalent if $\sLi$ is at most linear in the string-localized 
 potentials, as is the case in QED.  
 
 Alternative (\refeq{eq:alternative_S}) is desirable if one wants to keep 
 the delocalization as small as possible. However, it is in general not realizable. 
 If there is only a single string variable, the kernels (\refeq{eq:kernel_general_s}) 
 of the propagators of the involved string-localized potentials need 
 to be pulled back to the $e$-diagonal and consequently will contain 
 ill-defined products $u_+ \cdot u_-$. Now, the Hörmander criterion 
 Eq.~(\refeq{eq:hoermander_krit}) is only sufficient but not necessary, 
 meaning that it is not fully excluded that one can make sense of 
 $u_+ \cdot u_-$ although the Hörmander criterion is not 
 met.\footnote{The result may then have unwanted properties such as 
 that the Leibniz rule is not applicable.} In the case at hand, 
 however, the divergence can be observed explicitly \cite{GRT21}. This 
 rules out alternative (\refeq{eq:alternative_S}) for spacelike 
 strings.
 
 \begin{remk}
 There have been approaches that employ alternative (\refeq{eq:alternative_S}) 
 for \textit{lightlike} string variables and massive string-localized 
 potentials \cite{GBMV18}. But such an approach comes with other drawbacks, 
 which we will describe in Appendix \ref{app:lightlike_strings}.
 \end{remk}
 
 The interaction Lagrangian $\sLi$ depends only on a single 
 $x$-variable. One can therefore argue that alternative 
 (\refeq{eq:alternative_L}) is a natural choice to set up 
 perturbation theory in SLFT. In this case, loop graph contributions 
 would consist of products of propagators in $x$ and 
 $e$ and one must expect renormalization to become very complicated. 
 Recent observations in the string-localized equivalent of massless 
 Yang-Mills theory suggest that alternative (\refeq{eq:alternative_L}) 
 does not reproduce the standard model of particle physics 
 \cite{GGBM21}.\footnote{The cited work considers our alternative 
 (\refeq{eq:alternative_A}) from the beginning without mentioning other 
 alternatives but one can verify without too much effort 
 that the Lie algebra structure of gluon self-interactions is not compatible 
 with alternative (\refeq{eq:alternative_L}) by adjusting Section 2.3 
 in \cite{GGBM21} according to $\sLi(x,e_1,e_2)\rightarrow \sLi(x,e)$.}
 This observation rules out alternative (\refeq{eq:alternative_L}) for 
 phenomenological reasons.

 We are thus left with the alternative (\refeq{eq:alternative_A}), 
 i.e., $\sLi=\sLi(x,e_1,\cdots,e_k) =: \sLi(x,\stringarray)$, which 
 is also employed in the cited work \cite{GGBM21}. The analyses 
 therein additionally require a symmetry under exchange of all 
 string variables that appear in a fixed order of perturbation theory. 
 This symmetry can be achieved by smearing all string variables with 
 the same averaging function $c\in\sD(H)$ with 
 $  \int d^4e \, c(e) =1 $.\footnote{The test function $c$ needs 
 to have integral equal to $1$ if the string-localized potential is 
 to remain a potential for the field strength after smearing out 
 the $e$-variable.}
 With this at hand, we are finally 
 able to write down a candidate for the string-localized S-matrix, 
 \begin{align}
 \bS[g;c] := 1 + \sum_{n=1}^\infty \frac{i^n}{n!} \prod_{j=1}^n 
 \prod_{l=1}^k \int d^{4}x_j \int d^4e_{j,l} \, 
 g(x_j) c(e_{j,l}) \,S_n(x_1,\stringarray_1;\dots;x_n,\stringarray_n),
  \label{eq:dyson-series} 
\end{align}
 where the first-order coupling $S_1(x,\stringarray) 
 = \wick:\sLi(x,\stringarray):$ is the Wick-ordered interaction 
   Lagrangian. The property that $c$ integrates to unity ensures consistency if 
   $\sLi$ is a sum of terms where different powers of string-localized 
   potentials appear. 
  The higher-order couplings $S_n$ need to be constructed recursively 
  as time-ordered products of the first-order coupling. 
 
 However, the construction of time-ordered products in a string-localized  
 field theory is a non-trivial task, for one needs to make sense of how to 
 order several semi-infinite strings in time. 
 An axiomatic framework comparable to the one that is available 
 in point-localized QFT has not yet been formulated. 
 
  One approach towards a construction of time-ordered products  
 in string-localized QFT is called \textit{string-chopping} 
 \cite{CardosoMV18}. It proceeds by time-ordering segments 
 of string integrals wherever possible and taking account of 
 the singularity structure where time-ordering is ambiguous.  
 String chopping has been implemented for certain models 
 \cite{CardosoMV18,GGBM21} but a proof of its general validity 
 has not yet been given and it is still unclear how a generalization could 
 work. 
 
 The formulation of a fully self-contained and comprehensive axiomatic 
 framework for the construction of the time-ordered products in 
 Eq.~(\ref{eq:dyson-series}) is beyond the scope of this paper. 
 Instead, our aim is to show that there exist finite solutions 
 (in the sense of operator-valued distributions) for the string-localized 
 $S_n$ if they exist in the point-localized equivalent. 
 To prove the assertion, we write down the formal Wick-expansion  
 \begin{align}
  S_n
  &= T[\wick:\sLi(x_1,\stringarray_1):\dots\wick:\sLi(x_n,\stringarray_n):] 
  \notag \\
  &= \wick:\sLi(x_1,\stringarray_1)\dots\sLi(x_n,\stringarray_n): 
    \notag \\
  &\quad 
  + \sum_{\phi,\chi} \wick:\frac{\del \sLi(x_1,\stringarray_1)}{\del\phi}
  \frac{\del \sLi(x_2,\stringarray_2)}{\del\chi}\dots\sLi(x_n,\stringarray_n):  \vev{T\phi\chi} + \dots 
  \label{eq:wick_expansion}\\
  &\quad  +\sum_{\substack{\phi_1,\phi_2,\\ \chi_1,\chi_2}}
  \wick:\frac{\del^2 \sLi(x_1,\stringarray_1)}{\del\phi_1\phi_2}
  \frac{\del^2 \sLi(x_2,\stringarray_2)}{\del\chi_1\chi_2}\dots\sLi(x_n,\stringarray_n):  \vev{T\phi_1\chi_1}\vev{T\phi_2\chi_2} + \dots \notag \\
  &\quad+\dots \notag
 \end{align}
 as a sum containing a priori ill-defined products of propagators, some of which 
 may be string-localized. 
 An important property of the Dyson series
 Eq.~(\refeq{eq:dyson-series}) is that each string-localized potential 
 comes with its own string variable. This property has the consequence 
 that the products of propagators in Eq.~(\refeq{eq:wick_expansion})
 are products \textit{only} in the $x$-variables but tensor products in 
 the string variables. It is therefore enough to regularize 
 Eq.~(\ref{eq:wick_expansion}) \emph{after} integrating out the string 
 dependence of the propagators with the test function $c$. In Section 
 \ref{sec:WF}, we prove that, after smearing out the string variables, 
 the wavefront set of a relevant class of string-localized propagators 
 is contained in the wavefront set of the ordinary Feynman propagator. 
 As a consequence, the products of propagators appearing in 
 Eq.~(\ref{eq:wick_expansion}) are well-defined whenever they are 
 well-defined in the point-localized case and the regularization of 
 divergent amplitudes by extension of the products of propagators across 
 the remaining points stays a pure short distance problem in SLFT. 
 This statement 
 is due to the fact that the singularities of string-localized propagators, 
 which are the building blocks of the Wick expansion, are better behaved 
 than one might naively expect.

 Our considerations do not yield a full classification of the 
 Epstein-Glaser-like freedom of renormalization of the $S_n$. 
 However, because string-localized fields have an improved scaling 
 behavior compared to their point-localized counterparts, 
 these ambiguities are not expected to exceed the ones in 
 point-localized theories.

\section{Products of string-localized propagators}
\label{sec:WF}
 The heuristic considerations in Section 
 \ref{ssc:perturbation_theory} led us to the conclusion that each 
 string-localized potential has its own string variable and that all  
 string variables are smeared with the same test function. Therefore, 
 we can make a transition from Eq.~(\refeq{eq:def_slf_all_heli}) to 
 the smeared potentials 
 \begin{align}
  \label{eq:def_smeared_A}
  A_{c,\mu_1\cdots\mu_s}(x) := \int d^4e\,  
    A_{\mu_1\cdots\mu_s}(x,e) \,c(e)
 \end{align}
 for $c\in\sD(H)$ \textit{before} plugging them into the S-matrix. 
 Similarly, after using translation invariance of the propagator, 
 we can define a map 
 \begin{align}
  \label{eq:smeared_propagator}
  K: \sD(H) \rightarrow  \sS'(\mink),\; 
  c \mapsto \vev{T_0  A_{c,\mu_1\cdots\mu_s}  
    A_{c,\ka_1\cdots\ka_s}}(x).
 \end{align}
 The right-hand side of Eq.~(\refeq{eq:smeared_propagator}) is a 
 distribution only in $x$ but its singularities might depend on the 
 test function $c$. By Lemma \ref{lema:wf_kernel_smeared}, we have 
 \begin{align}
   \wf [ K(c)] \subset \set{ (x;p)  \;|\; &(x,e,e';p,0,0) 
    \in \wf \vev{T_0  A_{\mu_1\cdots\mu_s}(e)  
    A_{\ka_1\cdots\ka_s}(e')} , \notag \\ &e,e' \in \supp c
                            }.
   \label{eq:wf_smeared_dists}
 \end{align}
 The estimate (\refeq{eq:wf_smeared_dists}) is the key to  proving 
 that string-integration does not introduce new 
 singularities to the propagators of string-localized fields.
 
\subsection{Products of kinematic propagators}
\label{ssc:prod_kinematic_props}
We investigate the effect of the transition from distributions 
over $\mink \times H^2$ to distributions over $\mink$ on the 
kinematic string-localized propagators described in Section 
\ref{ssc:slf_props}, starting with the following lemma. 
\begin{lema}
\label{lema:qc_pm}
 For $c\in\sD(H)$, we define 
 \begin{align}
 \label{eq:def_qc_mu}
   q^{\mu_1\cdots\mu_s}_{c,\pm}(p):= \int d^4e \, c(e) 
   \frac{(\pm i)^s e^{\mu_1} \cdots e^{\mu_s}}{[(pe) \pm i0]^s}.
 \end{align}
 Then $q^{\mu_1\cdots\mu_s}_{c,\pm}(p) \in \sS'(\mink)$ with 
 $\wf q^{\mu_1\cdots\mu_s}_{c,\pm}(p) 
 =\set{(0;\lambda e) \;|\; \la \lessgtr 0, e \in \supp c}$.
 \end{lema}
\begin{proof}
 The expressions $q^{\mu_1\cdots\mu_s}_{c,\pm}(p)$ are the results of 
 smearing distributions of the form appearing in Corollary 
 \ref{corl:powers_upm} times a smooth (tensor-valued) function in the 
 string variable. Therefore, they are well-defined distributions. 
 By homogeneity, they are also tempered.
 
 Since $e\in H$ is spacelike and hence non-zero, the wavefront set of 
 $q^{\mu_1\cdots\mu_s}_{c,\pm}(p)$ must be contained in 
 $\set{ (p;x) \;|\; (p,e;x,0) \in \wf {u_\pm} }$
 by Lemma \ref{lema:wf_kernel_smeared},
 with $u_\pm$ as in Corollary \ref{corl:powers_upm} and 
 $\wf u_\pm$ as in  Eq.~(\ref{eq:wf_upm_spacelike}). This yields
 \begin{align}
  (p;x) \in \wf q^{\mu_1\cdots\mu_s}_{c,\pm}(p) 
  \quad\Rightarrow\quad 
  p=0 \text{ and } x=\la e  \text{ for some } \la\lessgtr 0 
  \text{ and } e\in\supp c.
 \end{align}
 To show that any such element $(0;\la e)$ is in the wavefront set, 
 note that the Fourier transform of $q^{\mu_1\cdots\mu_s}_{c,\pm}(p)$ 
 is 
 \begin{align}
  \int d^4p \, e^{i(px)} q^{\mu_1\cdots\mu_s}_{c,\pm}(p) 
  \sim \int d^4 e \, c(e)\, e^{\mu_1}\cdots e^{\mu_s} I_{\pm e}^s\delta(x)
 \end{align}
 with support $\set{x = \la e \;|\; \la\lessgtr 0,\; e \in \supp c}$. By 
 homogeneity and Lemma \ref{lema:wf_homogeneous}, 
 $(0;x)$ is an element of $\wf q^{\mu_1\cdots\mu_s}_{c,\pm}(p)$ if and only if $x$
 is in the support of the Fourier transform. This proves the claim.
\end{proof}
With Lemma \ref{lema:qc_pm} at hand, it is straightforward to adjust the proofs for the 
existence of the string-integrated kinematic propagators given in 
Section \ref{ssc:slf_props} to the expressions which are smeared 
in the string variables. When $p\ne0$, the expressions 
$ q^{\mu_1\cdots\mu_s}_{c,\pm}(p)$ are smooth and therefore, they 
can at most contribute to the wavefront set over $p=0$. We hence 
arrive at the following statement. 
\begin{lema}
 \label{lema:wf_kinematic_props_momentum_space}
 The kernels  
 \begin{align}
  \label{eq:FT_kinematic_props}
  \sF\left[  \vev{T_0  A^{(s)}_{c,f}  
   A^{(s)}_{c,f'}}(x) \right](p) =\int d^4e \int d^4 e' 
    \,c(e) c(e') \frac{{_mM}^{A^{(s)}_f,A^{(s)}_{f'}}(p,e,e')}{p^2-m^2+i0}
 \end{align}
 of the smeared kinematic string-localized propagators for all 
 masses $m\ge0$ and all spins respectively helicities $s\in\bN$ are tempered 
 distributions with 
 \begin{align}
  \label{eq:wf_smeared_kinematic_props_momentum_space}
   \wf \left(\sF\left[  \vev{T_0  A^{(s)}_{c,f}  
   A^{(s)}_{c,f'}}(x)  \right](p)\right)
  \subset
  \wf \frac{1}{p^2-m^2+i0} \cup \dot{T}^\ast_0.
 \end{align}
\end{lema}
In the massless case, the Fourier transform (\refeq{eq:FT_kinematic_props}) 
is homogeneous. Therefore, the wavefront set of the massless kinematic 
propagator in configuration space can be determined easily from 
Eq.~(\refeq{eq:wf_smeared_kinematic_props_momentum_space}) by use of 
Lemma \ref{lema:wf_homogeneous}. We obtain our first main result.
\begin{thm}[massless case]
 \label{thm:wf_massless}
 At $m=0$, the wavefront set of the smeared string-localized kinematic 
 propagator 
 \begin{align}
  \vev{T_0  A^{(s)}_{c,f} A^{(s)}_{c,f'}}(x)
   = \int \frac{d^4p}{(2\pi)^4} e^{-i(px)}\int d^4e \int d^4 e' 
    \,c(e) c(e') \frac{{_mM^{A^{(s)}_f,A^{(s)}_{f'}}}(p,e,e')}{p^2+i0}
 \end{align}
 is contained in the wavefront set of the massless point-localized 
 Feynman propagator from Eq.~(\refeq{eq:wf_DF}). In particular, products 
 of massless string-localized kinematic propagators and their product 
 with the propagators of point-localized fields are well-defined 
 on $\mink\setminus0$.
\end{thm}

In the massive case, homogeneity is lost and a transition from 
momentum to configuration space needs more effort. Before proving a 
similar statement to Theorem \ref{thm:wf_massless} for $m>0$, we 
prove an auxiliary lemma.

\begin{lema}
 \label{lema:wf_prod_p0}
  Let $u,v\in\sS'(\mink)$. Suppose further that $\hat{u}$ is polynomially 
  bounded, that $\wf \hat{u}\subset\dot{T}^\ast_0$,
  that $\hat{v}$ is smooth on a neighborhood of $p=0$ and that the 
  Hörmander product $\hat{u}\hat{v}$ is an element of $\sS'(\mink)$.
  Then $\wf \left[ \sF^{-1}(\hat{u}\hat{v})\right] \subset \wf v$.
\end{lema}
\begin{proof}
 We have to investigate the decay properties of the Fourier transforms of
 $\phi \sF^{-1}(\hat{u}\hat{v})$ for $\phi\in C_c^\infty(\mink)$. 
 Since $\phi$ is compactly supported and smooth, we know that its 
 Fourier transform is a Schwartz function, $\hat{\phi} \in \sS(\mink)$.
 Moreover, $\sF^{-1}(\hat{u}\hat{v})$ is a tempered distribution since 
 by assumption also $\hat{u}\hat{v}$ is tempered. Then 
 \begin{align}
  \left[ \sF\left( \phi \sF^{-1}(\hat{u}\hat{v}) \right) \right](p) 
  = \hat{\phi}\ast \hat{u}\hat{v}(p)
  = \hat{u}\hat{v}(\hat{\phi}(p-\cdot))
 \end{align}
is smooth and polynomially bounded \cite[Thm. IX.4]{RS75}. 

To investigate the decay of $\hat{u}\hat{v}(\hat{\phi}(p-\cdot))$, 
we introduce a second cutoff function $\chi\in C_c^\infty(\mink)$ 
with $0\le\chi\le 1$, $\chi\equiv 1$ on $B_r(0)$ and 
$\supp\chi\subset B_R(0)$, where $B_\varrho(p)$ is the closed ball 
of \textit{Euclidean} radius $\varrho$ and center $p\in\mink$, and 
where $0<r<R$ such that $B_R(0)\cap \singsupp \hat{v}=\emptyset$.

Then $\hat{u}$ is smooth on $\supp (1-\chi)$, $\hat{v}$ is smooth on $\supp \chi$ 
and $\hat{u}\hat{v} = \chi \hat{u}\hat{v} + (1-\chi) \hat{u}\hat{v}$.
The first term is unproblematic, for there are constants $N$, $C$ and $C'$ such 
that 
\begin{align}
 |\chi \hat{u}\hat{v} (\hat{\phi}(p-\cdot))| 
 &\le C \sum_{|\alpha+\beta|\le N} \sup_{k\in B_R(0)} 
    \left| k^\alpha \del_k^\beta \hat{\phi}(p-k) \right| \notag \\
 &\le C' (1+|p|)^N \sum_{|\alpha+\beta|\le N} \sup_{k\in B_R(p)} 
    \left| k^\alpha \del_k^\beta \hat{\phi}(k) \right|,
 \label{eq:estimate_lemma}
\end{align}
where $|p|$ is the Euclidean norm of $p$.
The righthand-side of Eq.~(\refeq{eq:estimate_lemma}) is rapidly 
decaying since $\hat{\phi}\in\sS(\mink)$ and since the supremum 
is taken over a compact set around $p$.

To estimate the second term $(1-\chi) \hat{u}\hat{v}$, note that the smooth function 
$(1-\chi)\hat{u}$ is polynomially bounded and thus
\begin{align}
 \left[(1-\chi) \hat{u}\hat{v}\right]\left(\hat{\phi}(p-\cdot)\right) 
 = \hat{v} \left((1-\chi) \hat{u}(\cdot) \,\hat{\phi}(p-\cdot)\right)
\end{align}
falls off rapidly if $\hat{v}\left(\hat{\phi}(p-\cdot)\right)=\widehat{\phi v}$ 
falls off rapidly. Thus also the full expression falls off rapidly if 
$\widehat{\phi v}$ does, which proves the lemma. 
\end{proof}

The desired statement about the wavefront sets of string-localized propagators 
for $m>0$ then follows as a corollary of Lemma \ref{lema:wf_prod_p0}:

\begin{thm}[massive case]
 \label{thm:wf_massive}
 At $m>0$, the wavefront set of the smeared string-localized kinematic 
 propagator 
  \begin{align}
  \vev{T_0  A^{(s)}_{c,f} A^{(s)}_{c,f'}}(x)
   = \int \frac{d^4p}{(2\pi)^4} e^{-i(px)}\int d^4e \int d^4 e' 
    \,c(e) c(e') \frac{{_mM^{A^{(s)}_f,A^{(s)}_{f'}}}(p,e,e')}{p^2-m^2+i0}
 \end{align}
 is contained in the wavefront set of the massive point-localized 
 Feynman propagator. In particular, 
 products of massive 
 string-localized kinematic propagators, their product with massless 
 string-localized kinematic propagators and with the propagators 
 of point-localized fields are well-defined on $\mink\setminus0$.
\end{thm}
\begin{proof}
 We define the distributions 
 \begin{align}
 \hat{u}(p) = \int d^4e \int d^4 e' \,c(e) c(e') 
 \,{_mM^{A^{(s)}_f,A^{(s)}_{f'}}}(p,e,e')
 \end{align}
 and $\hat{v}(p)=[p^2-m^2+i0]^{-1}$ with $m>0$. $\hat{u}(p)$ 
 is homogeneous of degree $0$ in $p$ and arises from contraction 
 of the distributions $q_{c,\pm}^{\mu_1\cdots\mu_s}$ from Lemma 
 \ref{lema:qc_pm} with a polynomial in $p$. Local integrability 
 of $\hat{u}(p)$ at $p=0$ ensures its existence as a tempered distribution 
 and by Lemma \ref{lema:qc_pm}, $\wf \hat{u} \subset \dot{T}^\ast_0$. 
 $\hat{v}$ is smooth at $p=0$ because $m>0$ and hence $\hat{u}$ and $\hat{v}$ 
 satisfy the assumptions of Lemma \ref{lema:wf_prod_p0}. Therefore, 
 Theorem \ref{thm:wf_massive} is a special case of 
 Lemma \ref{lema:wf_prod_p0} and 
 \begin{align}
  \wf{\vev{T_0  A^{(s)}_{c,f} A^{(s)}_{c,f'}}} 
  \subset \wf{v} 
  = \wf \left[\int \frac{d^4p}{(2\pi)^4} \frac{e^{-i(px)}}{p^2-m^2+i0} \right],
 \end{align}
 where $v$ is the massive scalar Feynman propagator, whose wavefront 
 set is the same as for the massless scalar Feynman propagator given by 
 Eq.~(\refeq{eq:wf_DF}) \cite{BrouderDH14}.
\end{proof}

\begin{remk}
Note that Lemma \ref{lema:wf_prod_p0} is only helpful at $m>0$ since 
wavefront set of the massless kernel 
$ [p^2+i0]^{-1}$ contains $\dot{T}^\ast_0$, so that the 
Hörmander product $\hat{u}\hat{v}$ of the respective 
$\hat{u}$ and $\hat{v}$ does not exist at $p=0$. For the same reason, 
there is no straightforward generalization of Lemma \ref{lema:wf_prod_p0} 
to the massless case.
\end{remk}

Theorems \ref{thm:wf_massless} and \ref{thm:wf_massive} show that 
the problem of regularizing divergent products of string-localized propagators 
is not posed worse than in point-localized QFT, provided that kinematic 
propagators are employed.
However, the transition from the two-point functions 
to the propagators  might not be unique as we have argued in Section
\ref{ssc:slf_props}. In the following section, we will show that 
Theorems \ref{thm:wf_massless} and \ref{thm:wf_massive} are 
generalizable to a large class of propagators.

\subsection{Products of non-kinematic propagators}
\label{ssc:non_kinematic_propagators}
Dependent on the scaling behavior at $x=0$ 
(see e.g.~\cite{BrunettiF00,ScharfLast} for the corresponding 
power counting arguments), there arise ambiguities in defining 
the propagator for the point-localized spin/helicity-$s$ field 
strength, 
\begin{align}
 \label{eq:ambiguities_TFF}
 \vev{T F^{(s)}_{\times} F^{(s)}_{\times}} 
 -\vev{T_0 F^{(s)}_{\times} F^{(s)}_{\times}} 
 = \sum_{|a|\le 2(s-1) } C^a_\times \del_\times^{|a|} \delta(x),
\end{align}
with a generic time-ordering recipe $T$, the kinematic time-ordering 
$T_0$, some constants $C^a_\times$ and 
where $\times$ is a placeholder for possible Lorentz indices. 

In order to not lose the connection between field strength and 
string-localized potential given by Eq.~(\refeq{eq:def_slf_all_heli}), 
it is a natural requirement that the freedom of choosing a time-ordering 
recipe for the string-localized potentials arises from the freedom 
(\refeq{eq:ambiguities_TFF}) by appropriate string-integration of the 
righthand-side. Then
\begin{align}
\label{eq:ambiguities_TAA_generic}
 \vev{T A^{(s)}_{c,\times} A^{(s)}_{c,\times}}
 - \vev{T_0 A^{(s)}_{c,\times} A^{(s)}_{c,\times}} 
 = \sum_{|a|\le 2(s-1) } C^a_\times  
   I_{c,+\times}^s I_{c,-\times}^s \del_\times^{|a|} \delta(x),
\end{align}
where we have introduced the smeared $s$-fold string-integration 
\begin{align}
 \label{eq:I_c_delta}
 I_{c,\pm \mu_1\cdots\mu_s}^s f(x)
 :=  \int d^4e \, c(e)
    \, e_{\mu_1}  \cdots  e_{\mu_s}  
    \, I_{\pm e}^s f(x).
\end{align}
The righthand-side of Eq.~(\refeq{eq:ambiguities_TAA_generic}) 
is a tempered distribution because it is a linear combination of 
Fourier transforms of 
\begin{align}
 (-ip)^{|a|}_\times 
  q_{c,+}^{\mu_1\cdots\mu_s}(p) q_{c,-}^{\ka_1\cdots\ka_s}(p),
\end{align}
 provided that the latter is locally integrable at $p=0$, i.e., if 
$|a|>2s-4$. 

Thus, the scaling behavior of the field strengths gives an 
upper bound on $|a|$, while the requirement for local 
integrability in momentum space gives a lower bound. 
In the massless case, the two-point functions of field 
strength and string-localized potential are homogeneous. 
The requirement that the propagators are homogeneous of 
the same degree then restricts the freedom to $|a|= 2s-2$ at 
$m=0$. In summary, we demand 
\begin{subequations}
\begin{align}
 \label{eq:ambiguities_TAA_massless}
 \vev{T A^{(s)}_{c,\times} A^{(s)}_{c,\times}}(x)
 -
  \vev{T_0 A^{(s)}_{c,\times} A^{(s)}_{c,\times}}(x)
 &= \sum_{|a|= 2s-2 } C^a_\times \del_\times^{|a|} 
   I_{c,\times}^s I_{-c,\times}^s \delta(x) 
   \word{at} m=0, \\
    \vev{T A^{(s)}_{c,\times} A^{(s)}_{c,\times}}(x)
    -
  \vev{T_0 A^{(s)}_{c,\times} A^{(s)}_{c,\times}}(x) 
 &= \sum_{|a|= 2s-3 }^{ 2s-2} C^a_\times \del_\times^{|a|} 
   I_{c,\times}^s I_{-c,\times}^s \delta(x) 
   \word{at} m>0  
    \label{eq:ambiguities_TAA_massive}
\end{align}
\end{subequations}
 for a general time-ordering recipe $T$ ordering string-localized 
 potentials. 
 \begin{remk}
 \label{eq:IR_and_homogeneity_constraint}
  We want to stress that the requirements $|a|=2s-2$ for $m=0$ and 
  $2s-3\le|a|\le 2s-2$ for $m>0$ are stronger than the power counting
  constraints. The latter would only imply $|a|\le 2s-2$ but 
  the integrability condition -- an \textit{infrared effect} -- 
  forbids all $|a|\le 2s-4$. At $m=0$, homogeneity gives an even 
  stronger constraint and excludes all $|a|<2s-2$.
 \end{remk}

 For time-ordering recipes $T$ that are subject to 
 Eq.s~(\refeq{eq:ambiguities_TAA_massless}) and 
 (\refeq{eq:ambiguities_TAA_massive}), the renormalization problem 
 is the same as for $T_0$:
 
\begin{thm}
 \label{thm:non_kinematic_does_not_increase_wf}
 Let $T$ denote a time-ordering recipe that  is subject to 
 Eq.~(\refeq{eq:ambiguities_TAA_massless}) if $m=0$ and 
 Eq.~(\refeq{eq:ambiguities_TAA_massive}) if $m>0$. Then 
 \begin{align}
  \wf\vev{T A^{(s)}_{c,\times} A^{(s)}_{c,\times}} 
  \subset \wf \vev{T_0 A^{(s)}_{c,\times} A^{(s)}_{c,\times}} 
 \end{align}
 and consequently, products of 
 $\vev{T A^{(s)}_{c,\times} A^{(s)}_{c,\times}}$ as well as 
 their products with propagators of point-localized fields 
 exist on $\mink\setminus0$.
\end{thm}
\begin{proof}
 By the constraints on $|a|$ coming from 
 Eq.s~(\refeq{eq:ambiguities_TAA_massless}) and 
 (\refeq{eq:ambiguities_TAA_massive}), the difference 
 between $\vev{T A^{(s)}_{c,\times} A^{(s)}_{c,\times}}$ and 
 $ \vev{T_0 A^{(s)}_{c,\times} A^{(s)}_{c,\times}}$ is a well-defined 
 tempered distribution since it is the Fourier transform of a sum 
 of distributions 
 \begin{align}
  \label{eq:p_qp_qm}
  (-ip)^{|a|}_\times 
 q_{c,+}^{\mu_1\cdots\mu_s}(p) q_{c,-}^{\ka_1\cdots\ka_s}(p),
 \end{align}
 which are locally integrable at $p=0$ and homogeneous in $p$. 
 Note that $q_{c,\pm}$ are smooth when $p\ne0$ by Lemma \ref{lema:qc_pm} 
 so that there are no issues with well-definedness of the expression 
 (\refeq{eq:p_qp_qm}). Moreover, the same lemma gives 
 \begin{align}
  \wf\left[ (-ip)^{|a|}_\times 
 q_{c,+}^{\mu_1\cdots\mu_s}(p) q_{c,-}^{\ka_1\cdots\ka_s}(p) \right] 
 \subset \dot{T}^\ast_0
 \end{align}
 so that homogeneity implies 
 \begin{align}
 \label{eq:wf_stringy_delta_smeared}
  \wf \left[ \del_\times^{|a|} 
 I_{c,\mu_1\cdots\mu_s}^s I_{-c,\ka_1\cdots \ka_s}^s \delta(x)  \right] 
 = \dot{T}^\ast_0
 \end{align}
 by Lemma \ref{lema:wf_homogeneous}, where the equality comes from the fact 
 that the Fourier transform (\refeq{eq:p_qp_qm}) is supported everywhere.
 Therefore, since the wavefront set of the sum of two distributions is 
 contained in the union of their wavefront sets, adding a linear 
 combination of derivatives of smeared string-deltas subject to 
 Eq.s~(\refeq{eq:ambiguities_TAA_massless}) or 
 (\refeq{eq:ambiguities_TAA_massive}), respectively, does not affect the 
 wavefront set over $\mink \setminus0$. Because the kinematic propagators 
 of the field strengths are derivatives of a fundamental solution of 
 the wave equation, $\dot{T}^\ast_0$ is already contained in their 
 wavefront set and hence also in the wavefront set of the kinematic 
 propagators of the string-localized potentials since the 
 wavefront set  of a string-integrated smeared Dirac delta is given by 
 Eq.~(\refeq{eq:wf_stringy_delta_smeared}).
\end{proof}

To summarize, Theorem \ref{thm:non_kinematic_does_not_increase_wf} 
gives a generalization of
Theorems \ref{thm:wf_massless} and \ref{thm:wf_massive} to all 
propagators that 
\begin{enumerate}
 \item arise from one of the field strength propagators displayed in 
        Eq.~(\refeq{eq:ambiguities_TFF}) by appropriate 
        string-integration, and 
 \item are subject to the constraints of power counting,   
       integrability in momentum space and, at $m=0$, homogeneity 
       of the same degree as the two-point function. 
\end{enumerate}
Note that only the lower bounds on $|a|$ are needed in the proof 
of Theorem \ref{thm:non_kinematic_does_not_increase_wf}, but 
not the constraints coming from power counting. The latter are 
only an additional requirement in order to reduce the (finite)
renormalization freedom. 

\begin{remk}
 In all our considerations in Sections \ref{ssc:prod_kinematic_props}
 and \ref{ssc:non_kinematic_propagators}, we have only considered 
 pure string-localized propagators 
 $\vev{T A^{(s)}_{c,\times} A^{(s)}_{c,\times}}$, whereas also mixed 
 propagators like $\vev{T F^{(s)}_{\times} A^{(s)}_{c,\times}}$ 
 are non-vanishing in SLFT. However, it should be obvious that such 
 propagators are subject to similar statements as 
 Theorems \ref{thm:wf_massless}, \ref{thm:wf_massive} and 
 \ref{thm:non_kinematic_does_not_increase_wf}.
\end{remk}

\section{Discussion}
\label{sec:discussion}
We have established that the regularization of a priori ill-defined 
products of propagators remains a pure short distance 
problem in SLFT, provided that 
each string-localized potential comes with its own string variable 
and provided that the pertinent propagators differ from the kinematic 
propagators only by string-localized Dirac deltas. 
Let us discuss some consequences of these results.

\subsection{A first step towards an Epstein-Glaser construction in SLFT}
\label{ssc:scaling_behavior}
Since each potential comes with its own string variable, their 
propagators can (and should) be smeared in the string variables before inserting 
them into the Dyson series for the scattering operator. As a result, 
the distributions appearing in the Dyson series only depend on the 
spacetime variable (and the smearing function for the string variables, 
of course). 

We proved in Sections \ref{ssc:prod_kinematic_props} and 
\ref{ssc:non_kinematic_propagators} that the wavefront set of smeared  
string-localized propagators is contained in the wavefront set of 
the ordinary point-localized Feynman propagator. 
Therefore, the problem of regularizing divergent loop graph amplitudes 
remains an extension problem across an $x$-diagonal. In particular, 
string integration does not introduce new singularities to the propagators 
after smearing out the string variables.
In other words, the divergences in SLFT are of the same nature as in 
point-localized approaches: \textit{they are pure UV divergences}. 

Due to the improved ultraviolet scaling behavior of the string-localized 
potentials, coming from the string integrations as in Eq.~(\ref{eq:def_slf_all_heli}), 
the freedom of choosing a regularization of divergent products of 
propagators is reduced compared to point-localized theories. 
It therefore seems a promising task to investigate whether 
one can formulate renormalizable string-localized models involving 
higher spin fields where the point-localized equivalent is non-renormalizable. 
A prime example may be the graviton self-coupling. However, to carry out 
such a task, a comprehensive axiomatic framework for the construction of 
$\bS[g;c]$ from Eq.~(\ref{eq:dyson-series}) needs to be set up. 
Only then can one hope to give a full classification of the ambiguities of 
a time-ordering prescription in SLFT.

The ultraviolet behavior -- or correspondingly: power counting -- 
is not the only way to constrain the renormalization freedom. 
Typical additional requirements in point-localized theories are 
that the renormalized time-ordered products must not be in 
conflict with gauge invariance \cite{ScharfLast} or that a refined 
characterization of Lorentz covariance is preserved during renormalization 
\cite{NST14}. Such notions can also serve as guiding principles 
for a full axiomatic construction of the S-matrix in SLFT, 
the proper substitute for gauge invariance being the principle 
of string independence (see for example \cite{VGB16} and \cite{GGBM21} 
for applications).

\subsection{Renormalization in practice}
\label{ssc:renorm_in_practice}
 No new singularities are introduced to string-localized propagators 
 if one sets up perturbation theory as described in this paper, meaning 
 that each string-localized potential in the Dyson series comes with its own 
 string variable. Thus, 
 there is no need to develop new renormalization techniques that are 
 specifically adjusted to SLFT. One can rely on well-known methods such 
 as analytic regularization or differential renormalization (see 
 for example \cite{Duetsch18} for an introduction) 
 to construct special extensions of products of propagators in SLFT.
 This is a remarkable statement, for analytic structures in SLFT are 
 quite complicated and this complexity is commonly considered to be 
 one of the main drawbacks of SLFT.
 
 \subsection{Connections to axial gauges}
 \label{ssc:axial_gauges}
 The analytic structure of propagators in axial or lightcone gauges 
 is similar to the one of string-localized propagators (see for 
 example \cite{Leibbrandt87} for an introduction to axial gauges). 
 Axial gauge propagators, however, are usually not treated as 
 distributions in the variable $n$, which represents the preferred 
 direction and is the analogue of the string variable in SLFT. Hence, 
 the singularities that arise when the Minkowski 
 product of $n$ with the momentum $p$ vanishes are of a different 
 nature than the ones discussed in this paper. 
 The singularities at $(pn)=0$ were an important reason for 
 the decreasing interest in axial gauges over the past decades.
 
 Adjusting the framework of axial gauges by treating the respective 
 propagators as distributions in $n$ and letting each appearing axial 
 gauge field depend on its own $n$, our results can be transferred 
 with benefit to axial gauge theories (with $n$ spacelike).  
 Thus, the singularities at $(pn)=0$ in axial gauges do not 
 cause additional problems for renormalization if they are 
 treated as described in this paper.
 
 Axial gauges suffer from analytic complexity but also offer advantages, 
 in particular if each axial gauge field comes with its own $n$:
 They prove useful in the so-called spinor-helicity formalism that 
 drastically reduces the computational effort to determine gluon 
 scattering matrix elements \cite[Chapters~25.4.3 and 27]{Schwartz14}. 
 Due to the close formal connection between axial gauge and 
 string-localized potentials, it is worthwhile to investigate whether 
 the spinor-helicity formalism can be adjusted to the string-localized 
 setup of perturbation theory presented in this paper.

 \subsection{Connections to the method of string-chopping}
 \label{ssc:string_chop}
 It is a non-trivial question how the time-ordered products of 
 string-localized fields -- or interaction Lagrangians -- can be 
 defined. In Section \ref{ssc:perturbation_theory}, we have mentioned  
 the method of string-chopping, first described by Cardoso, Mund and 
 V\'arilly \cite{CardosoMV18} for linear fields and later adjusted 
 to certain models containing self-interactions of string-localized 
 potentials \cite{GGBM21}. In a nutshell, string-chopping says that 
 the strings appearing at each order in perturbation theory can 
 be chopped into a finite number of compact segments plus an infinite 
 tail for each string, so that all pieces can be meaningfully ordered 
 in time. Moreover, the result of the time-ordering arising from 
 string-chopping is unique outside an exceptional set, which consists  
 of the configurations where two of the appearing strings 
 intersect. A generalization of that method to arbitrary 
 models has not yet been proven but seems a natural conjecture.

 Morally, the formal Wick expansion Eq.~(\ref{eq:wick_expansion}) 
 is a realization of the string chopping method because the kinematic propagator 
 of the string-localized potentials from Theorem 
 \ref{thm:well_def_kinematic_props_all_s}, obtained by inserting 
 the string-localized kernel (\refeq{eq:kernel_general_s}) into the 
 kinematic propagator of point-localized fields 
 Eq.~(\refeq{eq:kinematic_prop_pl}), automatically chops the strings 
 and orders them in time. This can be seen by recalling that 
 it is defined as the string-integral over propagators which are 
 already time-ordered with respect to their arguments. 
 Therefore, string-integration promotes the step functions 
 $\theta(\pm x^0)$, which are responsible for time-ordering 
 in the point-localized propagators, to 
 \begin{align}
 \theta(\pm (x^0+te^0-t'{e'}^0)) 
 \end{align}
 so that an automatic chopping is implicitly achieved. 
 Similar to the point-localized case described by Epstein and Glaser 
 \cite{EpsteinGlaser73}, causality as in Eq.~(\refeq{eq:slf_CR}) 
 implies that the propagator is uniquely defined outside the set 
 \begin{align}
 \set{x+te-t'e'=0\text{ for some } t,t'\ge0},
 \end{align}
 which is in accordance with the ambiguities from 
 Eq.s~(\refeq{eq:ambiguities_TAA_massless}) and 
 (\refeq{eq:ambiguities_TAA_massive}) as well as with 
 the ambiguities observed in the abstract formulation of string 
 chopping \cite{CardosoMV18,GGBM21}.

 
 \subsection*{Acknowledgements}
  The author is grateful to K.-H. Rehren,  
  J. M. Gracia-Bond\'ia and J. Mund for fruitful discussions and comments 
  and thanks the reviewer for valuable suggestions. He received 
  financial support from the Studienstiftung des deutschen Volkes e.V.
  
  
  \appendix
  
 \section{Other choices of string variables}
 Throughout this paper, we have always worked with spacelike string 
 variables $e$ living in the open subset $H=\set{e^2<0}\subset\mink$. 
 In the literature, one also finds other choices: lightlike string 
 variables, normalized spacelike string variables with Minkowski 
 square $e^2=-1$ or purely 
 spacelike string variables $e=(0,\vec{e})$, all of which correspond 
 to restrictions of the string variables to \textit{closed} subsets
 (or more precisely, closed submanifolds). 
 Such restrictions are much more subtle than the restriction to $H$ 
 used in this paper. We briefly examine the described options. 
 
 \subsection{Lighlike strings} 
  \label{app:lightlike_strings}
 Lightlike string directions have been employed in \cite{GBMV18} when 
 dealing with massive string-localized potentials, where they 
 promise a computational advantage. The authors of \cite{GBMV18} were 
 able to set equal all string variables appearing in the Dyson series 
 for the scattering operator describing the weak interaction by 
 exploiting that the problematic denominator 
 $\left([(pe)+i0][(pe)-i0]\right)^{-1}$ in $E_{\mu\kappa}(p,e,e')|_{e=e'}$ 
 from Eq.~(\refeq{eq:def_Emunu}) drops out when $e^2=0$. This 
 simplification of $E_{\mu\kappa}$ yields an essential reduction of the 
 complexity of tree-graph calculations. Similarly, one can check that 
 also the problematic terms in the kernel for $s=2$ given by 
 Eq.~(\refeq{eq:kernel_general_s}) drop out, resulting in an even bigger 
 computational simplification than for $s=1$.\footnote{It is a conjecture 
 of the author that the problematic denominators drop out for any helicity 
 but whether that happens or not is of no interest for our current 
 considerations.}
 
 However, the authors of \cite{GBMV18} restricted their considerations 
 to tree graph contributions, where no products (or convolution products 
 in momentum space) of several $E_{\mu\kappa}(p,e,e')|_{e=e'}$ appear. 
 It is very likely that this changes when treating loop graph contributions 
 and therefore, the divergent denominators will pop up again in loop 
 amplitudes, resulting in complex renormalization schemes and  
 spoiling the computational advantage that was achieved at tree level. 
  
 One can also think of SLFT with lightlike strings where not all 
 string variables are set equal. However, an analysis similar to 
 the one presented in this paper cannot be performed in that case. 
 This is due to the fact that the restriction to the closed set of 
 lightlike string directions causes trouble. Without loss of 
 generality, we can investigate the restriction to lightlike string 
 variables with zero-component equal to $1$, 
 which is given by the pullback of the respective inclusion map 
 \cite[Corollary 8.2.7.]{Hoermander90}, provided that this pullback exists. Thus, consider the map 
 \begin{align}
 \label{eq:def_iota_lightlike}
  \iota: \mink \times (0,2\pi) \times (0,\pi) \rightarrow (\mink)^2, \;
  (p,\varphi,\vartheta) \mapsto (p,e),
  \text{ where } e
  =\begin{pmatrix}
  1 \\ \sin\vartheta \cos\varphi \\ \sin\vartheta \sin\varphi \\ \cos\vartheta
  \end{pmatrix},
 \end{align}
 so that the desired restriction is the pullback $\iota^\ast U_\pm$ with 
 $U_\pm$ as in Lemma \ref{lema:Upm_welldef_wf}.
 \begin{remk}
  The submanifold of elements $(p,e)$, where $e$ is lightlike and has 
  $0$-component equal to one is $\mink\times \bS^2$. To avoid confusion 
  with coordinate-related singularities, one needs several charts. 
  The map $\iota$ corresponds to only a single chart but is enough to 
  demonstrate the issues that come with lightlike strings. 
 \end{remk}

 Having a look at Lemma \ref{lema:Upm_welldef_wf} and using 
 \begin{align}
  {^t\iota'} \begin{pmatrix}
              \xi \\ \eta
             \end{pmatrix}
 = \begin{pmatrix}
    \xi \\ \sin\vartheta (\eta_1 \sin \varphi - \eta_2 \cos\varphi) \\
    \eta_3 \sin\vartheta  - \cos\vartheta(\eta_1 \cos \varphi + \eta_2 \sin\varphi)
   \end{pmatrix}
 \end{align}
 for $(\xi,\eta)\in(\mink)^2$, one can easily verify that
 $\iota^\ast U_\pm$ is well-defined but that the wavefront set of 
 the pullback contains elements $(p,\varphi,\vartheta;\lambda e,0,0)$ 
 when $p$ becomes proportional to $e$, the latter being defined as 
 in Eq.~(\refeq{eq:def_iota_lightlike}). Note that the 
 singular-support-criterion $(pe)=0$ is met when $p$ is proportional 
 to $e$ only if $e$ is lightlike (or $p=0$).

 Hence, there is no immediate analogue 
 of Lemmas \ref{lema:well_def_SI_2pf_and_prop_massless} and 
 \ref{lema:well_def_SI_2pf_and_prop_massive} for the case of lightlike 
 strings and in particular, analyses as performed in Section \ref{sec:WF}, 
 which led to a simple renormalization description, are not feasible 
 for lightlike strings because lightlike strings produce additional 
 singularities also when $0\ne p = \la e$. 
 This problem is worse in the massless case than in the massive case, 
 for $p^2+i0$ is singular when $p=\la e$, but $p^2-m^2+i0$ with $m>0$ 
 is not.
 
  In conclusion, spacelike strings seem preferable over lightlike 
 strings from analytic and heuristic viewpoints. Nevertheless, 
 lightlike strings cannot be fully excluded at the present time. 
 
 \subsection{Closed subsets of spacelike strings}
 \label{app:H_minus_1}
 In Remark \ref{remk:closed_spacelike}, we claimed that the 
 restriction to the closed submanifold $H_{-1}$ of normalized 
 spacelike string directions with Minkowski square $-1$ is harmless. 
 In principle, this restriction can cause similar issues as the 
 restriction to the lightlike string directions, but a brief 
 analysis shows that it is indeed much better behaved than the 
 latter. Similar to the case of lightlike strings, we consider 
 an inclusion map 
 \begin{align}
 \notag
  &\tilde{\iota}: \mink \times \bR \times (0,2\pi) \times (0,\pi)
  \rightarrow (\mink)^2, \\
  &(p,\tau,\varphi,\vartheta) \mapsto (p,e),
  \text{ where } e
  =\begin{pmatrix}
  \sinh\tau \\ \cosh\tau \sin\vartheta \cos\varphi 
  \\  \cosh\tau\sin\vartheta \sin\varphi \\  \cosh\tau\cos\vartheta
  \end{pmatrix},
   \label{eq:def_iota_H_minus_1}
 \end{align}
 which is again only a single chart but a generalization to cover 
 the full submanifold is straightforward. For $(\xi,\eta)\in(\mink)^2$, 
 we have 
 \begin{align}
 \label{eq:pullback_H_minus_1_derivative}
  {^t\tilde{\iota}\,'} \begin{pmatrix}
              \xi \\ \eta
             \end{pmatrix}
 = \begin{pmatrix}
    \xi \\ 
     \eta^0 \cosh\tau - \sinh\tau 
     \left[ \eta_1 \sin\vartheta \cos\varphi 
        + \eta_2 \sin\vartheta \sin\varphi 
        + \eta_3 \cos\vartheta\right]\\
    \cosh\tau \sin\vartheta (\eta_1 \sin \varphi - \eta_2 \cos\varphi)\\
    \cosh\tau\left[ \eta_3 \sin\vartheta 
    - \cos\vartheta(\eta_1 \cos \varphi + \eta_2 \sin\varphi)\right] 
   \end{pmatrix}
 \end{align}
  and thus, the pullback $\tilde{\iota}^\ast U_\pm$ is well-defined 
  by Lemmas \ref{lema:pullbacks} and \ref{lema:Upm_welldef_wf}. In 
  contrast to the case of lightlike string variables, the wavefront 
  set of the pullback does not contain 
  elements $(p,\tau,\varphi,\vartheta; \la e,0,0,0)$, provided that 
  $p\ne0$. This can be seen by inserting $\eta=\la p$, $\la\ne0$, into 
  Eq.~(\refeq{eq:pullback_H_minus_1_derivative}) and noting that the 
  pullback is only singular when 
  \begin{align}
   (pe)= p^0 \sinh\tau - \cosh\tau 
     \left[ p_1 \sin\vartheta \cos\varphi 
        + p_2 \sin\vartheta \sin\varphi 
        + p_3 \cos\vartheta\right]
     =0.
  \end{align}
 Consequently, the results in Sections \ref{sec:slf} and 
 \ref{sec:WF} remain valid also if one restricts to $H_{-1}$. 
 We nevertheless chose to consider the restriction to the open set 
 $H$ in the main part of the paper because it is much simpler and also 
 exhibits the practical advantage that one can easily derive with 
 respect to the string variables.

 \begin{remk}
  \label{remk:H_minus_1_homogeneity}
  A qualitative and simpler argument that the restriction to $H_{-1}$ 
  is unproblematic is the homogeneity in the string-variables of all 
  string-localized propagators of degree $\omega=0$: 
  When one interprets $H$ as $H_{-1}\times \bR_{\ge0}$, the ``radial'' 
  part is constant and can simply be integrated out with the radial part 
  of the test function.
 \end{remk}

 \subsection{Purely spacelike strings}
 \label{app:purely_sl}
 Another case appearing in the literature \cite{MundRS20} is the case 
 of purely spacelike string variables $e=(0,\vec{e})$, for example with 
 $|\vec{e}|=1$. It is motivated by the fact that the inner product 
 $-(ee')$ becomes positive definite, which is not the case in $H$ or 
 $H_{-1}$. This case can be investigated by adjusting the inclusion 
 map (\refeq{eq:def_iota_lightlike}) from the lightlike case by setting 
 the zero-component of $e$ to $0$ instead of $1$. Then the only -- but 
 very important -- difference in the wavefront set analysis is the 
 criterion for the singular support, which becomes 
 \begin{align}
  \vec{p}\cdot \vec{e}=0 \word{instead of} \vec{p}\cdot \vec{e}=p^0
  \word{in the lightlike case.}
 \end{align}
 The wavefront set of the restriction to purely 
 spacelike strings can then only contain elements 
 $(p,\varphi,\vartheta; \la e, 0,0)$ if $\vec{p}=0$ but not when 
 $p=\kappa e$ for some $\ka \ne 0$, in contrast to the lightlike 
 case. Nevertheless, the appearance of these critical elements 
 in the wavefront set 
 can happen for arbitrary $p^0$ and hence, our results from 
 Sections \ref{sec:slf} and \ref{sec:WF} cannot be directly 
 transferred to a restriction to purely spacelike strings. 
 
 However, in the mentioned application \cite{MundRS20}, a 
 time-ordering of the string-localized expression is not required
 because there, the string-localized part is perturbed with a 
 point-localized Lagrangian.

\bibliographystyle{plain}

\footnotesize
\begin{thebibliography}{10}

\bibitem{AsteS99}
A.~Aste and G.~Scharf.
\newblock {Non-abelian gauge theories as a consequence of perturbative quantum
  gauge invariance}.
\newblock {\em Int. J. Mod. Phys. A}, \textbf{14}:3421--3434, 1999.

\bibitem{BogoliubovS59}
N.~N. Bogoliubov and D.~V. Shirkov.
\newblock {\em {Introduction to the Theory of Quantized Fields}}.
\newblock Interscience, 1959.

\bibitem{BrouderDH14}
C.~Brouder, N.~V. Dang, and F.~H\'elein.
\newblock {A smooth introduction to the wavefront set}.
\newblock {\em Journal of Physics A: Mathematical and Theoretical},
  \textbf{47}:443001, 2014.

\bibitem{BrunettiF00}
R.~Brunetti and K.~Fredenhagen.
\newblock {Microlocal Analysis and Interacting Quantum Field Theories:
  Renormalization on Physical Backgrounds}.
\newblock {\em Commun. Math. Phys.}, \textbf{208}:623–661, 2000.

\bibitem{BF82}
D.~Buchholz and K.~Fredenhagen.
\newblock {Locality and the structure of particle states}.
\newblock {\em Commun. Math. Phys.}, \textbf{84}:1--54, 1982.

\bibitem{CardosoMV18}
L.~T. Cardoso, J.~Mund, and J.~C. V\'arilly.
\newblock {String chopping and time-ordered products of linear string-localized
  quantum fields}.
\newblock {\em Math. Phys. Anal. Geom.}, \textbf{21}:3, 2018.

\bibitem{Dang16}
N.V. Dang.
\newblock {The Extension of Distributions on Manifolds, a Microlocal Approach}.
\newblock {\em Ann. Henri Poincaré}, \textbf{17}:819–859, 2016.

\bibitem{Dirac55}
P.~A.~M. Dirac.
\newblock {Gauge-invariant formulation of quantum electrodynamics}.
\newblock {\em Canadian Journal of Physics}, \textbf{33}:650--660, 1955.

\bibitem{Duetsch18}
M.~Dütsch.
\newblock {\em {From classical field theory to perturbative quantum field
  theory}}.
\newblock Birkhäuser, Basel, 2018.

\bibitem{DuetschS99}
M.~Dütsch and G.~Scharf.
\newblock {Perturbative gauge invariance: the electroweak theory}.
\newblock {\em Ann. Phys. (Leipzig)}, \textbf{8}:359--387, 1999.

\bibitem{EpsteinGlaser73}
H.~Epstein and V.~J. Glaser.
\newblock {The role of locality in perturbation theory}.
\newblock {\em Ann. Inst. Henri Poincaré A}, \textbf{19}:211--295, 1973.

\bibitem{GGBM21}
C.~Gaß, J.~M. Gracia-Bondía, and J.~Mund.
\newblock {Revisiting the Okubo-Marshak argument}.
\newblock {\em Symmetry}, 13(9), 2021.

\bibitem{GRT21}
C.~Gaß, K.-H. Rehren, and F.~Tippner.
\newblock {On the spacetime structure of infrared divergencies in QED}.
\newblock {\em Lett. Math. Phys.}, \textbf{112}:37, 2022.

\bibitem{GBMV18}
J.~M. Gracia-Bondía, J.~Mund, and J.~C. Várilly.
\newblock {The chirality theorem}.
\newblock {\em Ann. Henri Poincaré}, \textbf{19}:843--874, 2018.

\bibitem{Hoermander90}
L.~Hörmander.
\newblock {\em {The Analysis of Linear Partial Differential Operators I}}.
\newblock Springer, Berlin, 2nd edition, 1990.

\bibitem{Jordan35}
P.~Jordan.
\newblock {Zur Quantenelektrodynamik. III. Eichinvariante Quantelung und
  Diracsche Magnetpole}.
\newblock {\em Zeitschrift fur Physik}, \textbf{97}:535–537, 1935.

\bibitem{Leibbrandt87}
G.~Leibbrandt.
\newblock {Introduction to noncovariant gauges}.
\newblock {\em Rev. Mod. Phys.}, \textbf{59}:1067--1119, 1987.

\bibitem{Mandelstam62}
S.~Mandelstam.
\newblock {Quantum electrodynamics without potentials}.
\newblock {\em Ann. Phys. (NY)}, \textbf{19}:1--24, 1962.

\bibitem{MunddeO17}
J.~Mund and E.T. de~Oliveira.
\newblock {String-Localized Free Vector and Tensor Potentials for Massive
  Particles with Any Spin: I. Bosons}.
\newblock {\em Commun. Math. Phys.}, \textbf{355}:1243–1282, 2017.

\bibitem{MundRS21}
J.~Mund, K.-H. Rehren, and B.~Schroer.
\newblock {Infraparticle quantum fields and the formation of photon clouds}.
\newblock {\em JHEP}, \textbf{04}:083, 2022.

\bibitem{MundRS20}
J.~Mund, K.-H. Rehren, and B.~Schroer.
\newblock {Gauss’ Law and string-localized quantum field theory}.
\newblock {\em JHEP}, \textbf{01}:2020, 001.

\bibitem{MundRS17b}
J.~Mund, K.-H. Rehren, and B.~Schroer.
\newblock {Helicity decoupling in the massless limit of massive tensor fields}.
\newblock {\em Nucl. Phys. B}, \textbf{924}:699--727, 2017.

\bibitem{MundRS17a}
J.~Mund, K.-H. Rehren, and B.~Schroer.
\newblock {Relations between positivity, localization and degrees of freedom:
  the Weinberg--Witten theorem and the van Dam--Veltman--Zakharov
  discontinuity}.
\newblock {\em Phys. Lett. B}, \textbf{773}:625--631, 2017.

\bibitem{MundSY04}
J.~Mund, B.~Schroer, and J.~Yngvason.
\newblock {String-localized quantum fields from Wigner representations}.
\newblock {\em Phys. Lett. B}, \textbf{596}:156--162, 2004.

\bibitem{MundSY06}
J.~Mund, B.~Schroer, and J.~Yngvason.
\newblock {String-localized quantum fields and modular localization}.
\newblock {\em Commun. Math. Phys.}, \textbf{268}:621--672, 2006.

\bibitem{NST14}
N.~M. Nikolov, R.~Stora, and I.~Todorov.
\newblock {Renormalization of massless Feynman amplitudes in configuration
  space}.
\newblock {\em Rev. Math. Phys.}, \textbf{26}:1430002, 2014.

\bibitem{RS75}
M.~Reed and B.~Simon.
\newblock {\em {Methods of modern mathematical physics II}}.
\newblock Academic Press, San Diego, 1975.

\bibitem{RehrenPLLimit}
K.-H. Rehren.
\newblock {Pauli--Lubański limit and stress-energy tensor for infinite-spin
  fields}.
\newblock {\em JHEP}, \textbf{11}:130, 2017.

\bibitem{ScharfLast}
G.~Scharf.
\newblock {\em {Gauge Field Theories: Spin One and Spin Two}}.
\newblock Dover, New York, 2016.

\bibitem{Schroer19}
B.~Schroer.
\newblock {The role of positivity and causality in interactions involving
  higher spin}.
\newblock {\em Nucl. Phys. B}, \textbf{941}:91--144, 2019.

\bibitem{Schwartz14}
M.~D. Schwartz.
\newblock {\em {Quantum Field Theory and the Standard Model}}.
\newblock Cambridge University Press, Cambridge, 2014.

\bibitem{Steinmann71}
O.~Steinmann.
\newblock {\em {Perturbation Expansions in Axiomatic Field Theory}}.
\newblock Lect. Notes in Phys. \textbf{11}, Springer, Berlin, 1971.

\bibitem{Steinmann82}
O.~Steinmann.
\newblock {A Jost--Schroer theorem for string fields}.
\newblock {\em Commun. Math. Phys.}, \textbf{87}:259--264, 1982.

\bibitem{Steinmann84}
O.~Steinmann.
\newblock {Perturbative QED in terms of gauge invariant fields}.
\newblock {\em Ann. Phys. (NY)}, \textbf{157}:232--254, 1984.

\bibitem{vDV70}
T.~van Dam and M.~Veltman.
\newblock {Massive and massless Yang--Mills and gravitational fields}.
\newblock {\em Nucl. Phys. B}, \textbf{22}:397--411, 1970.

\bibitem{VZ69}
G.~Velo and D.~Zwanziger.
\newblock {Noncausality and other defects of interaction Lagrangians for
  particles with spin one and higher}.
\newblock {\em Phys. Rev.}, \textbf{188}:2218, 1969.

\bibitem{VGB16}
J.~C. Várilly and J.~M. Gracia-Bondía.
\newblock {Stora's fine notion of divergent amplitudes}.
\newblock {\em Nuclear Physics B}, \textbf{912}:28--37, 2016.

\bibitem{Weinberg95I}
S.~Weinberg.
\newblock {\em {The Quantum Theory of Fields I}}.
\newblock Cambridge University Press, Cambridge, 1995.

\bibitem{WW80}
S.~Weinberg and E.~Witten.
\newblock {Limits on massless particles}.
\newblock {\em Phys. Lett. B}, \textbf{96}:59--62, 1980.

\bibitem{Zh70}
V.~I. Zakharov.
\newblock {Linearized graviton theory and the graviton mass}.
\newblock {\em JETP Lett.}, \textbf{12}:312--313, 1970.

\end{thebibliography}

\end{document}